\documentclass[a4paper,onecolumn,11pt,accepted=2026-11-17]{quantumarticle}
\pdfoutput=1
\usepackage[utf8]{inputenc}
\usepackage[english]{babel}
\usepackage[T1]{fontenc}
\usepackage{amsmath}
\usepackage[dvipsnames,svgnames,x11names,hyperref]{xcolor}
\usepackage{amsthm}
\usepackage{soul}
\usepackage{comment}
\usepackage{thm-restate}
\usepackage{thmtools}
\makeatletter
\@ifundefined{newcounteralias}{}{%
  \renewcommand\thmt@autorefsetup{%
    \@xa\def\csname\thmt@envname autorefname\@xa\endcsname\@xa{\thmt@thmname}%
  }%
}
\makeatother
\usepackage[final,colorlinks=true,breaklinks=true]{hyperref}
\hypersetup{linkcolor=quantumviolet,citecolor=quantumviolet,urlcolor=quantumviolet}
\usepackage{amssymb,amsfonts,amsmath,amsthm,amscd,dsfont,mathrsfs}
\usepackage{complexity}
\usepackage{tikz}
\usetikzlibrary{decorations.pathreplacing,backgrounds,calc}
\usepackage{tikz-cd}
\usepackage{multirow}
\usepackage{braket}
\usepackage[capitalise,nameinlink]{cleveref}
\usepackage{xcolor}
\usepackage{caption}

\usepackage{booktabs, multirow}
\usepackage[utf8]{inputenc}
\usepackage[english]{babel}
\usepackage[T1]{fontenc}
\usepackage{amsmath}

\usepackage[backend=biber,style=alphabetic,minalphanames=3,maxalphanames=4,maxnames=99,backref=true,doi=true,url=true,hyperref=true]{biblatex}

\usepackage{url}
\usepackage{bbm}
\usepackage{xurl}
\usepackage{etoolbox}

\appto{\bibsetup}{\sloppy}

  \theoremstyle{plain}
  \newtheorem{theorem}{Theorem}[section]
  \newtheorem{lemma}[theorem]{Lemma}

  \newtheorem{corollary}[theorem]{Corollary}
  \newtheorem{proposition}{Proposition}
  \newtheorem{definition}[theorem]{Definition}

    \newtheorem{fact}[theorem]{Fact}

 \newtheorem*{theorem*}{Theorem}

\crefname{informallemma}{Informal Lemma}{Informal Lemmas}
\Crefname{informallemma}{Informal Lemma}{Informal Lemmas}
\crefname{informaltheorem}{Informal Theorem}{Informal Theorems}
\Crefname{informaltheorem}{Informal Theorem}{Informal Theorems}
\crefname{claim}{Claim}{Claims}
\Crefname{claim}{Claim}{Claims}
\crefname{fact}{Fact}{Facts}
\Crefname{fact}{Fact}{Facts}
\crefname{importedtheorem}{Imported Theorem}{Imported Theorems}
\Crefname{importedtheorem}{Imported Theorem}{Imported Theorems}
\crefname{construction}{Construction}{Constructions}
\Crefname{construction}{Construction}{Constructions}

\AddToHook{env/theorem/begin}{\crefalias{section}{theorem}}
\AddToHook{env/lemma/begin}{\crefalias{section}{lemma}}
\AddToHook{env/informallemma/begin}{\crefalias{section}{informallemma}}
\AddToHook{env/informaltheorem/begin}{\crefalias{section}{informaltheorem}}
\AddToHook{env/corollary/begin}{\crefalias{section}{corollary}}
\AddToHook{env/proposition/begin}{\crefalias{section}{proposition}}
\AddToHook{env/definition/begin}{\crefalias{section}{definition}}
\AddToHook{env/remark/begin}{\crefalias{section}{remark}}
\AddToHook{env/claim/begin}{\crefalias{section}{claim}}
\AddToHook{env/fact/begin}{\crefalias{section}{fact}}
\AddToHook{env/example/begin}{\crefalias{section}{example}}
\AddToHook{env/importedtheorem/begin}{\crefalias{section}{importedtheorem}}
\AddToHook{env/construction/begin}{\crefalias{section}{construction}}

\definecolor{DarkGreen}{RGB}{0,150,0}

\newcommand{\DarkGreen}[1]{\textcolor{DarkGreen}{#1}}

\newcommand{\blue}[1]{{\color{blue}#1}}
\newcommand{\red}[1]{{\color{red}#1}}

\newcommand{\gap}{\mathrm{gap}}
\renewcommand{\deg}{\mathrm{deg}}

\newcommand{\proj}[1]{\ket{#1}\!\bra{#1}}

\renewcommand{\C}{\mathbb{C}}

\newcommand{\eps}{\varepsilon}
\newcommand{\kLH}{$\gamma$-$2$-$\mathrm{LH}$}

\renewcommand{\cP}{\mathcal{P}}

\newcommand{\snp}{{(\mathrm{snap})}}

\makeatletter
\AtBeginDocument{%
  \fancyfoot[L]{%
    \href{https://quantum-journal.org/?s=Rounding\%20Almost\%20Commuting\%20Hamiltonians\&reason=title-click}{%
      \usebox{\@quantumacceptedbox}%
    }%
  }%
}
\makeatother

\begin{document}

\title{Rounding Almost Commuting Hamiltonians}

\author{Islam Faisal}
\email{islam@bu.edu}
\affiliation{Boston University, Boston, MA, USA}
\author{Anand Natarajan}
\email{anandn@mit.edu}
\thanks{partially supported by NSF CAREER Award 2339948.}
\affiliation{Massachusetts Institute of Technology, Cambridge, MA, USA}
\author{Alexander Poremba}
\email{poremba@bu.edu}
\affiliation{Boston University, Boston, MA, USA}
\maketitle

\begin{abstract}
Commuting Hamiltonians lie at the boundary between classical constraint satisfaction and quantum many-body physics, exhibiting rich quantum structure while remaining more tractable than general noncommuting models.  In contrast, physical Hamiltonians are rarely exactly commuting,
which motivates the study of \emph{almost commuting} Hamiltonians. Despite their relevance, the implications of approximate commutation are only poorly understood.

In this work, we show how to efficiently approximate any almost commuting $2$-local qubit Hamiltonian by a commuting one: we give a new locality-preserving \emph{algorithmic rounding technique} that maps any $2$-local Hamiltonian $H=\sum_{i=1}^m h_i$ with $\|[h_i,h_j]\| \leq \eps$ to a nearby Hamiltonian
$\hat{H}$ whose terms pair-wise commute, and which is within overall distance
$\|H-\hat{H}\| = O(m\,\eps^{1/3})$.

As a consequence, we show that $\delta$-approximations to the ground energy for
$\eps$-almost commuting $2$-local qubit Hamiltonians
lie in $\mathsf{NP}$ when $\delta \gg m\eps^{1/3}$, extending the classical containment well
beyond the commuting setting.   Finally, we present two applications of our rounding framework: Gibbs sampling and fast Hamiltonian simulation for almost commuting systems.
\end{abstract}

\section{Introduction}
\label{sec:introduction}
Commuting local Hamiltonians have been studied ~\cite{bravyi2004commutativeversionklocalhamiltonian,aharonov2011complexity,schuch2011complexitycommutinghamiltonianssquare,AKV,irani2023commutinglocalhamiltonianproblem,bostanci2025commutinglocalhamiltonians2d} across condensed matter physics, quantum information theory, and computational complexity. From the perspective of constraint satisfaction, commuting local Hamiltonians sit at a boundary between classical and quantum problems.  When all local terms commute, the ground-state problem reduces to the simultaneous satisfaction of a family of local constraints~\cite{bravyi2004commutativeversionklocalhamiltonian}, analogous to a classical constraint satisfaction problem (CSP). Indeed, many commuting Hamiltonians admit ground states that can be described efficiently in \emph{classical} terms, such as product states, stabilizer states, or tensor network states of bounded bond dimension~\cite{Gharibian_2015}. The computational complexity of commuting Hamiltonians is thus significantly more tractable than that of general noncommuting models, with many cases lying in the classical complexity class $\mathsf{NP}$~\cite{bravyi2004commutativeversionklocalhamiltonian,irani2023commutinglocalhamiltonianproblem,bostanci2025commutinglocalhamiltonians2d}.
In condensed matter physics, commuting Hamiltonians serve as exactly solvable models that capture the essential structure of entire phases of matter.  Prominent examples include stabilizer Hamiltonians for quantum error-correcting codes (e.g., the toric code~\cite{Kitaev_2003}), commuting-projector models for topological order~\cite{Levin_2005}, and fracton models with constrained dynamics~\cite{Vijay_2016}. 
Formally, the $k$-local commuting Hamiltonian problem consists of Hamiltonians of the form
\begin{align}
H=\sum_{i=1}^m h_i \quad \text{ such that } \quad [h_i,h_j]=0, \quad\quad  \forall i,j=1,\dots,m \,
\end{align}
where all of the terms pair-wise commute\footnote{The \emph{commutator} of two operators $A$ and $B$ is defined as $[A,B] = AB - BA$.} and each term is of unit norm $\|h_i\|\leq 1$ acting on at most $k$ qubits.
Commuting Hamiltonians retain intrinsically quantum features that sharply distinguish them from classical CSPs: although all of the Hamiltonian terms $\{h_i\}_{i=1}^m$ can be simultaneously diagonalized, the common eigenbasis may not be a product basis, and the ground states may exhibit long-range entanglement and topological order, such as the toric code~\cite{Kitaev_2003}. From a complexity-theoretic perspective, this places commuting Hamiltonians in an intermediate regime between classical satisfiability and fully quantum local Hamiltonians, providing a test ground for difficult problems such as the quantum PCP conjecture~\cite{aharonov2011complexity,Hastings2012}, the NLTS conjecture~\cite{Anshu_2023}, Gibbs state preparation~\cite{kastoryano2016quantumgibbssamplerscommuting} and fast thermalization~\cite{chen2023fastthermalizationeigenstatethermalization}.

\paragraph{Almost commuting Hamiltonians.}
Despite their utility, commuting Hamiltonians represent a highly nongeneric limit.  Physical systems are rarely exactly commuting, and even infinitesimal noncommuting perturbations can produce dramatic qualitative changes. 
 This motivates the study of \emph{almost commuting} Hamiltonians, in 
 which local terms fail to commute but only do so weakly; for example, when the Hamiltonian terms pair-wise "$\eps$-almost commute" in the sense that
\begin{align}
\|[h_i,h_j]\| \leq \eps, \quad\quad \forall i,j=1,\dots,m.   
\end{align}
This class of Hamiltonians captures a much richer range of physical phenomena, including 
entanglement generation, slow dynamics, and emergent conservation laws~\cite{Sachdev2011,Mbeng_2024}. 
A simple and illustrative example is the transverse-field Ising model 
in the weak-field regime, described by the $2$-local qubit Hamiltonian
\begin{align}\label{eq:ising}
H = -\sum_{\langle i,j \rangle} Z_i Z_j - h \sum_i X_i \, ,
\end{align}
with $h \ll 1$.  At $h=0$, the Hamiltonian consists of exactly 
commuting terms and reduces to a classical Ising model with 
product-state ground states.  Turning on a \emph{weak} transverse field 
introduces local noncommuting terms whose pairwise commutator norms 
$\|[Z_i Z_j,hX_k]\|$
scale with $h$ for $i=k$ or $j=k$, resulting in an almost commuting Hamiltonian. Generally speaking, even
mild noncommutativity can cause the underlying physics to 
change qualitatively: the ground state becomes entangled, quantum fluctuations lift degeneracies, and the system supports nontrivial correlations which are absent in the commuting case~\cite{Mbeng_2024}.  This illustrates how almost commuting Hamiltonians interpolate between classical constraint satisfaction and genuinely quantum many-body behavior.

The weak transverse-field Ising model might suggest that any almost commuting Hamiltonian can always be identified as a simple perturbation of an underlying commuting model. However, this intuition can be misleading in general. While the perturbation theory~\cite{Kato:1966:PTL} perspective is useful in identifying approximately commuting structure in very special or integrable models, this relies on strong assumptions---such as a known commuting parent Hamiltonian or small, easily identifiable model parameters---which may not be present in generic local Hamiltonians. 

From a computational standpoint, almost commuting Hamiltonians occupy a compelling intermediate regime.  At first sight, they appear no longer classically reducible; and yet they are far from fully quantum Hamiltonians.  Understanding whether their ground states and low-energy properties retain some of the tractability of commuting models—or instead exhibit the full hardness of general quantum constraint satisfaction problems—remains an intriguing open question.

\paragraph{Making almost commuting matrices commute.}

A naive idea for handling almost commuting systems, in general, is to simply discard Hamiltonian terms until the remainder commute; for example, one could consider the anti-commutation graph whose vertices correspond to Hamiltonian terms and whose edges connect pairs that fail to commute, and then try to find a \emph{maximum independent set} (or some efficient approximation)~\cite{10.1007/BF01994876}. The resulting Hamiltonian, obtained by retaining only the corresponding subset of terms, is exactly commuting by construction. However, this approach does not exploit the \emph{magnitude} of the noncommutativity: its error is determined solely by how many terms are removed. In the worst case, even when all $\varepsilon$-almost commute for a small $\varepsilon >0$, the procedure may discard a constant fraction of the terms, leading to a vacuous approximation error that scales with the number of Hamiltonian terms $m$ with no improvement as $\varepsilon \to 0$. 
In contrast, a meaningful rounding procedure for almost commuting Hamiltonians should be \emph{$\varepsilon$-sensitive}, with an error that decreases as the commutators become smaller. This suggests that an $\varepsilon$-sensitive rounding procedure must do more than just \emph{pruning} a noncommuting interaction graph: it must exploit the algebraic structure of the Hamiltonian terms themselves.

A more sophisticated angle of attack for dealing with a set of \emph{almost commuting} matrices is to explicitly find nearby \emph{commuting} matrices; for example, starting with a pair
\begin{align}
A,B \in \mathbb{C}^{d \times d}\quad \text{ with } \quad \|[A,B]\| \leq \eps    
\end{align}
one might try to \emph{round} these matrices to $\hat{A},\hat{B}$ such that $[\hat A,\hat B]=0$ with $\|\hat A-A\| \leq \delta(\eps)$ and
$\|\hat B-B\| \leq \delta(\eps)$, for some $\delta(\eps)$ such that $\delta(\eps) \rightarrow 0$ as $\eps \rightarrow 0$.
The mathematical problem of rounding almost commuting matrices has a long and subtle history. For pairs of Hermitian matrices of norm at most one, Lin's theorem~\cite{Lin1997} guarantees a rounding bound $\delta(\eps)$ that is independent of the matrix dimension and tends to zero with $\eps$; quantitative bounds were later given by Hastings~\cite{Hastings_2009,herrera2022hastingsapproachlinstheorem}.
This dimension-independent conclusion fails for \emph{triplets}, as shown by Davidson~\cite{Davidson1985}: there are families of three Hermitian matrices of norm at most one, in increasing dimensions, whose pairwise commutator norms tend to zero while every commuting triple remains a fixed positive operator-norm distance away in at least one component. Hastings and Loring~\cite{Hastings_2010} describe related obstructions using the \emph{Bott index}, which can be nonzero for almost commuting triples, is stable under sufficiently small perturbations, and vanishes for commuting triples whenever it is defined. These results suggest that large collections of almost commuting Hermitian matrices can, in general, not be rounded to nearby commuting matrices unless further structure is enforced, hindering a direct application to almost commuting Hamiltonians. Nevertheless, many physically relevant non-commuting local Hamiltonians are semi-classical, and may very well be close to a commuting Hamiltonian, as the weak transverse-field Ising model in \Cref{eq:ising} suggests. This raises the question:
\begin{center}
\emph{Is it possible to round generic almost commuting Hamiltonians
to nearby commuting Hamiltonians?}
\end{center}

Progress on this question could shed new light on the extent to which weak noncommutativity is essential for the observed physics, versus merely a perturbation of an underlying commuting structure. Moreover, from a complexity perspective, it would allow us to better understand the boundary between classical satisfiability and inherently quantum optimization problems.

\subsection{Our contributions}

In this work, we overcome many prior limitations on rounding 
approximately commuting Hermitian matrices which have so far hindered 
applications to many-body Hamiltonians. We develop the first 
\emph{locality-preserving algorithmic rounding technique} that allows one to 
simultaneously
round large collections of almost commuting Hermitian matrices---avoiding the regimes where 
general no-go results apply---by exploiting the special structure of physical many-body Hamiltonians:
\begin{itemize}
\item (\emph{Locality:}) Interactions in physical many-body Hamiltonians tend 
to be \emph{local} in the sense that they involve a constant number of 
particles at a time; consequently, each term in the Hamiltonian is 
actually a 
highly non-generic Hermitian operator: it only acts non-trivially on a 
few sub-systems at a time and is equal to the identity everywhere else.

\item (\emph{Dimensionality:}) Many-body Hamiltonians tend to 
involve particles whose Hilbert space carries a fixed local dimension. 
In the case of two-level qubit systems (or "spin systems"), such constrained dimensionality can give rise to very useful properties, such as the fact that commutation between Hermitian single-qubit operators becomes 
\emph{transitive}: if $A,B$ commute and $B,C$ 
commute, then $A$ and $C$ must also commute, provided that $B$ is not 
a multiple of the identity.\footnote{We prove this statement in \Cref{lem:transitivity}.} In higher dimensions, this property can fail due to \emph{degenerate eigenspaces}.\footnote{For example, suppose that the matrices $A$ and $C$ act differently inside the same degenerate subspace of the matrix $B$. Then, both matrices can commute with $B$ but they need not commute with each other.}
\end{itemize}

Our rounding technique exploits these physical constraints, 
and allows us to
efficiently approximate any almost commuting $2$-local 
Hamiltonian by a nearby commuting one, while preserving locality.
This is captured by our main 
result, which we informally state below.

\begin{restatable}[Formally,~\Cref{thm:transforming-hamiltonian}]{informaltheorem}{InformalTheoremOne}
Let $H=\sum_{i=1}^m h_i$ be an $\eps$-almost commuting $2$-local 
Hamiltonian on $n$ qubits. Then, there exists a deterministic classical algorithm that, given as input a description of $H$ and the bound $\eps$, runs 
in linear time and outputs a nearby commuting 2-local Hamiltonian $\hat{H}=\sum_{i=1}^m 
\hat{h}_i$ such that $\|h_i - \hat{h}_i\| \leq O(\eps^{1/3})$, for all
$i \in [m]$, and which is within overall distance 
$\|\hat{H} - H\| \leq O(m \,\eps^{1/3})$.
\end{restatable}
The constants in these error bounds are independent of the maximum degree of the interaction graph: no bounded-degree assumption is needed. In particular, the error for each individual term is $O(\eps^{1/3})$, while summing these errors gives the overall $O(m\eps^{1/3})$ bound.

Our result is specific to qubits and does not establish a corresponding rounding theorem for general $2$-local qudit Hamiltonians. The proof uses the fact that every non-scalar Hermitian single-qubit operator has one-dimensional eigenspaces, so operators commuting with that pivot also commute with one another. In higher local dimension, a non-scalar pivot can have degenerate eigenspaces; pinching then makes operators block diagonal, but their restrictions within a degenerate block need not commute. Extending the method to qudits therefore requires additional ideas.

To the best of our knowledge, this yields the first \emph{algorithmic} rounding technique for almost commuting $2$-local qubit Hamiltonians which produces an explicit, quantitative error bound that vanishes as $\varepsilon \rightarrow 0$, unlike naive methods that fail to take the magnitude of the noncommutativity into account.\footnote{For a detailed comparison with the \emph{non}-constructive rounding in~\cite{arad2011notepartialnogotheorem}, see \Cref{sec:relatedwork}.}
As a direct consequence of our rounding result, we show the following:

\begin{restatable}[Formally, ~\Cref{corollary:NP-Containment}]{informaltheorem}{InformalContainment}
    The $\eps$-almost commuting $2$-local qubit Hamiltonian problem with $m$ terms and energy promise gap\footnote{Here we mean the difference in energy thresholds for the YES and NO cases of the local Hamiltonian problem. Note that in the formal statement of \Cref{corollary:NP-Containment} we use the \emph{relative} promise gap $\gamma$, which is equal to the energy promise gap divided by $m$.} $\delta$ is contained in $\NP$ when $\delta \gg m \eps^{1/3}$.
\end{restatable}
This extends the classical $\mathsf{NP}$ containment result of Bravyi and Vyalyi~\cite{bravyi2004commutativeversionklocalhamiltonian} well
beyond the commuting setting.
In addition, the above theorem also implies a new no-go result in the context of quantum PCPs~\cite{aharonov2011complexity,Hastings2012} akin to~\cite{arad2011notepartialnogotheorem}. Specifically, it suggests that any qPCP-hard family of $2$-local qubit Hamiltonians must exhibit some moderate amount of noncommutativity: assuming $\mathsf{NP} \neq \mathsf{QMA}$, any
family of $2$-local qubit Hamiltonians for which the $\gamma$-LH problem is $\QMA$-hard for \emph{constant} $\gamma$ must allow for pairwise commutator norms of $\Omega(1)$.

More broadly, this places almost commuting Hamiltonians in an intermediate regime between classical constraint satisfaction and fully quantum optimization. It shows that, at least for $2$-local qubit systems, genuine $\mathsf{QMA}$-hardness requires a significant departure from the structure of commuting models, thereby linking a complexity-theoretic transition to a quantitative, operator-algebraic measure of noncommutativity.

\begin{figure}[t]
\centering
\begin{tikzpicture}[
    scale=1.1,
    node distance=1.2cm,
    spin/.style={circle, draw=black, fill=white, inner sep=2pt},
    term/.style={rounded corners, draw=black, very thick},
    approx arrow/.style={->, very thick},
    bg lattice/.style={gray!40, thick},
    fg lattice/.style={black, very thick}
]

\begin{scope}[shift={(-3,0)}]

\node at (1.5,3.8) {$\|[h_i,h_j]\|\le \varepsilon$};

\foreach \x in {0,1,2,3} {
  \foreach \y in {0,1,2,3} {
    \node[spin] (s\x\y) at (\x,\y) {};
  }
}

\foreach \x in {0,1,2} {
  \foreach \y in {0,1,2,3} {
    \draw[fg lattice] (s\x\y) -- (s\the\numexpr\x+1\relax\y);
  }
}
\foreach \x in {0,1,2,3} {
  \foreach \y in {0,1,2} {
    \draw[fg lattice] (s\x\y) -- (s\x\the\numexpr\y+1\relax);
  }
}

\node[term, draw=blue!70, minimum width=1.6cm, minimum height=0.9cm]
  at (1.5,1.8) {};

\node[term, draw=red!70,
      minimum width=1.6cm, minimum height=1.1cm,
      rotate=90]
  at (2.0,1.45) {};

\node at (0.55,1.55) {$h_i$};
\node at (2.7,0.7) {$h_j$};

\node at (1.45,-0.4) {\small Almost commuting Hamiltonian};

\end{scope}


\draw[approx arrow] (1,1.8) -- (2,1.8);
\node at (1.5,2.3) {\small Rounding};

\begin{scope}[shift={(3,0)}]

\node at (1.5,3.8) {$[\hat h_i,\hat h_j]=0$};

\foreach \x in {0,1,2,3} {
  \foreach \y in {0,1,2,3} {
    \node[spin, fill=gray!20] (t\x\y) at (\x+0.15,\y) {};
  }
}

\foreach \x in {0,1,2} {
  \foreach \y in {0,1,2,3} {
    \draw[bg lattice] (t\x\y) -- (t\the\numexpr\x+1\relax\y);
  }
}
\foreach \x in {0,1,2,3} {
  \foreach \y in {0,1,2} {
    \draw[bg lattice] (t\x\y) -- (t\x\the\numexpr\y+1\relax);
  }
}

\node[term, draw=blue!70, fill=blue!5, fill opacity=0.2,
      minimum width=1.6cm, minimum height=0.9cm]
  at (1.65,1.8) {};

\node[term, draw=red!70, fill=red!5, fill opacity=0.2,
      minimum width=2.1cm, minimum height=1.1cm,
      rotate=90]
  at (2.15,1.55) {};

\node at (0.7,1.55) {$\hat h_i$};
\node at (2.9,0.7) {$\hat h_j$};

\node at (1.67,-0.4) {\small Nearby commuting Hamiltonian};

\end{scope}

\end{tikzpicture}
\caption{Illustration of our locality-preserving rounding transformation. 
Left: an almost commuting Hamiltonian $H=\sum_{i=1}^m h_i$ on a 2D lattice. 
Right: a nearby commuting Hamiltonian $\hat{H}=\sum_{i=1}^m \hat{h}_i$ of exactly the same locality. We remark that our rounding map does not require geometric locality.}
\end{figure}

Finally, in \cref{sec:lower-bound}, we give a simple example Hamiltonian consisting of $2$ terms that shows that a rounding error of $\Omega(\eps^{1/2})$ is necessary. We discuss the prospects for making our bounds tight below in \cref{sec:future}.

\subsection{Applications}

As an application of our rounding technique for almost commuting quantum many-body systems, we present two algorithmic applications; the first one in the context of quantum Gibbs sampling and the second in the context of Hamiltonian simulation.

\paragraph{Quantum Gibbs sampling.}

The Gibbs ensemble connects microscopic descriptions of physical systems with their macroscopic thermodynamic behavior, allowing one to compute quantities such as free energies, effective partition functions, and equilibrium expectation values~\cite{PathriaBeale2011,Chandler1987,Callen1985,LandauLifshitz1980,Sachdev2011,AltlandSimons2010,Mermin1965,Tuckerman2010,LandauBinder2014}. Given a local Hamiltonian $H$ and an inverse temperature $\beta > 0$, the Gibbs state is given by
\begin{align}\label{eq:gibbs-ex}
\rho_\beta(H) = \frac{e^{-\beta H}}{\mathrm{Tr}\left[e^{-\beta 
H}\right]}.
\end{align}
A recent line of work~\cite{Temme_2011,GilyenThermal23a,chen2023efficient,BCL24,rajakumar2024gibbs} has proposed \textit{quantum} algorithms for preparing such thermal states. Although these methods can offer theoretical speedups~\cite{BCL24,rajakumar2024gibbs}, their performance typically hinges on either mixing-related characteristics or temperature-dependent quantities that are challenging to estimate or control in realistic settings.

A prominent category of Gibbs sampling tasks---one which typically admits polynomial-time algorithms with provable guarantees~\cite{hwang2025gibbsstatepreparationcommuting,schmidhuber2025hamiltoniandecodedquantuminterferometry}---is the \emph{commuting case}, where $H$ is a \emph{commuting} Hamiltonian.
By leveraging the \emph{Structure 
Lemma}~\cite{bravyi2004commutativeversionklocalhamiltonian}, recent 
work~\cite{hwang2025gibbsstatepreparationcommuting} showed that Gibbs 
sampling for large classes of commuting Hamiltonians even reduces to 
\emph{Gibbs sampling} for certain \emph{classical} Hamiltonians; these 
in turn can admit powerful classical methods such as Swendsen-Wang 
dynamics~\cite{PhysRevLett.58.86,feng2022swendsenwangdynamicsferromagneticising}. While the commuting case has been widely studied, only little is known about the non-commuting case, let alone the \emph{almost commuting} case.\footnote{
We remark that Schmidhuber et 
al.~\cite{schmidhuber2025hamiltoniandecodedquantuminterferometry} 
recently studied another form of "almost commuting" Hamiltonians;
namely, signed Pauli Hamiltonians with a \emph{sparse} anti-commutation 
graph. In their work, they provide evidence that Gibbs sampling and 
optimization for such almost commuting 
systems is classically intractable, hinting at potential quantum 
advantage of the HDQI algorithm.
However, their notion of almost commuting Hamiltonians is quite different from 
ours; in particular, our rounding algorithm does not seem to apply in their setting.
}
Currently, all existing quantum Gibbs sampling methods fail to take 
such almost commuting structure into account; in particular, the 
techniques for the commuting case are known to break down in the presence of (even mild) non-commuting interactions. This raises the question: 
\emph{is it possible to extend quantum Gibbs sampling techniques for commuting Hamiltonians towards almost-commuting Hamiltonians?}
In \Cref{sec:Gibbs}, we give an affirmative answer to this question: we 
use our rounding technique to reduce Gibbs sampling for certain almost 
commuting $2$-local Hamiltonians to Gibbs sampling for a nearby 
commuting Hamiltonian, thereby allowing one to once again make use of 
the full toolbox for the commuting case~\cite{hwang2025gibbsstatepreparationcommuting,schmidhuber2025hamiltoniandecodedquantuminterferometry}.
Moreover, we show that a standard continuity bound on the stability of the Gibbs state suffices to ensure 
that the reduction is sound, and indeed approximately samples from the 
target state in \Cref{eq:gibbs-ex}.

\paragraph{Faster Hamiltonian simulation.}
Simulating quantum-mechanical systems is one of the central tasks in quantum algorithms.
It is well known that Trotter-based Hamiltonian simulation methods can perform better in the presence of commutativity in the Hamiltonian. In fact, for an exactly commuting Hamiltonian, the runtime of Hamiltonian simulation can be logarithmic in the simulation time $t$ (as explained in \Cref{lem:com-ham-fast}), unlike the linear or worse scaling in the general case.
This makes commutativity a useful algorithmic resource for real-time dynamics, not only a simplifying structural assumption for ground-state problems.
For Hamiltonians that do not exactly commute, prior investigations~\cite{childs2021theory} have shown that commutator bounds can lead to faster Trotter-based simulation algorithms.  Our setting is complementary: rather than analyzing the Trotter error directly, we first round the Hamiltonian to a nearby commuting one and then simulate the remaining discrepancy separately.

In \Cref{sec:hamiltonian-simulation}, we observe that our rounding, together with the "interaction picture" method of Low and Wiebe~\cite{low2018hamiltonian}, gives us a way to directly leverage the first-order commutators alone to speed up Hamiltonian simulation. In a nutshell, we apply the algorithm of~\cite{low2018hamiltonian} with the \emph{fast} part of the Hamiltonian being the commuting rounding, and the \emph{slow} part being the error term incurred in the rounding. 
Intuitively, the system is evolved mostly under a nearby exactly commuting Hamiltonian, while the genuinely noncommuting remainder is treated as a perturbation in the interaction picture.
Apart from a commuting-simulation overhead and polylogarithmic factors, the time-dependent cost is governed by a norm bound for the latter term, controlled by the commutator error,
rather than by the total norm of the Hamiltonian.
In this way, the simulation cost is governed by how far the Hamiltonian is from the commuting regime, rather than by the full norm of the Hamiltonian itself.

\subsection{Related work}
\label{sec:relatedwork}

The work of~\cite{overlapping} studied projections that pairwise nearly commute and used a predecessor (see Theorem 3.2~\cite{overlapping}) of the pinching technique we extend here. However, their techniques do not preserve the locality of the operators, and hence do not apply in our setting.

The specific question of rounding approximately commuting Hamiltonians to commuting Hamiltonians was first studied by Arad~\cite{arad2011notepartialnogotheorem}. In this work, it is shown that for two-local Hamiltonians on $d$-dimensional qudits on a $D$-regular interaction graph, whose terms are projectors, an "asymptotically good" rounding exists: that is, for fixed $d$ and $D$, in the limit of the commutator error $\eps$ going to $0$, the rounding error also goes to $0$. The rounding error may depend on both $d$ and $D$.  However, the rounding is \emph{non-constructive} and no explicit bound on the rounding error was established. Arad's techniques also do not easily lend themselves to algorithmic rounding, since they are based on a nonconstructive compactness argument inspired by a work of Halmos~\cite{halmos1976some}.
The rounding results in~\cite{arad2011notepartialnogotheorem} are incomparable to ours: they are weaker in that they only apply to Hamiltonians whose terms are projectors; at the same time, they are also more general in that they apply to local dimensions greater than $2$ as well.
We remark that Arad's proof also makes use of local decompositions of the Hamiltonian terms, similar to the Pauli decompositions that we exploit. Comparing our results to his, we consider our main contribution to be the locality-preserving algorithmic rounding theorem with an explicit polynomial bound on the rounding error given by our analysis---to our knowledge, the first constructive rounding technique with an explicit error bound in terms of $\epsilon$ that goes to zero as $\epsilon \rightarrow 0$.
In fact, coming up with such a bound was listed by Arad as an open question, with the suggestion that it might follow easily from the techniques of~\cite{pearcy1979almost}. This latter paper belongs to the rich literature on rounding approximately commuting high-dimensional matrices. We do not believe this literature is directly useful or relevant to our problem because we are interested specifically in roundings that preserve the \emph{local structure} of the Hamiltonian.

Turning to complexity theory, our work connects to the extensive literature
on commuting local Hamiltonians.  Beginning with the result of Bravyi and
Vyalyi~\cite{bravyi2004commutativeversionklocalhamiltonian} that commuting
$2$-local qubit Hamiltonians lie in $\NP$, subsequent works have extended
$\NP$ containment results to broader commuting settings, including
higher-locality qubit systems, certain qutrit systems, two-dimensional
geometries, and more general interaction complexes~\cite{aharonov2011complexity,schuch2011complexitycommutinghamiltonianssquare,AKV,irani2023commutinglocalhamiltonianproblem,bostanci2025commutinglocalhamiltonians2d}.
The central contribution of this work is to extend this viewpoint beyond
exact commutativity.  We show that sufficiently almost commuting $2$-local
qubit Hamiltonians can be rounded to nearby commuting Hamiltonians, thus
extending the classical containment well beyond the commuting setting.

\subsection{Technical Overview}
Let us now proceed with a technical overview of the result in this paper. We assume familiarity with the basic formalism of quantum information, and refer the reader to \Cref{sec:background} for a refresher.

To illustrate our rounding technique, we consider the simple (although already non-trivial\footnote{We remark that prior no-go results on rounding \emph{general} Hermitian matrices already apply to triplets~\cite{Davidson1985}.}) example of an $\eps$-almost-commuting $2$-local Hamiltonian on a three-qubit triangle,
\begin{align}
H = h_{1,2} + h_{1,3} + h_{2,3} \, , 
\end{align}
where the $i,j$ terms pair-wise $\epsilon$-almost commute such that 
\begin{align}
\| [h_{1,2},h_{1,3}]\| \leq \eps \, , \quad \| [h_{1,3},h_{2,3}]\| \leq \eps \quad\,\,\, \text{ and } \quad \| [h_{1,2},h_{2,3}]\| \leq \eps
\end{align}
for some $\eps \in (0,1)$. The goal is to identify an \emph{exactly commuting} Hamiltonian $\hat{H} = \hat h_{1,2} + \hat h_{1,3} + \hat h_{2,3}$ with $[\hat h_{1,2},\hat h_{1,3}] = [\hat h_{1,3},\hat h_{2,3}] = [\hat h_{1,2},\hat h_{2,3}] = 0$, which inherits the locality of the original Hamiltonian, and which is \emph{nearby} in the sense that
\begin{align}
\| h_{1,2} - \hat h_{1,2}\| \leq \delta(\eps) \, , \quad\| h_{1,3} - \hat h_{1,3}\| \leq \delta(\eps), \quad\,\,\, \text{ and } \quad \| h_{2,3} - \hat h_{2,3}\| \leq \delta(\eps)
\end{align}
for some function $\delta(\eps)$ with the property that $\delta(\eps) \rightarrow 0$ as $\eps \rightarrow 0$.

Let us first consider the Hamiltonian terms $h_{1,2}$ and $h_{1,3}$. 
Using a local expansion in terms of the Pauli basis $\{\sigma^{(i)}\}_{i=0}^3$ given in \Cref{algebra-fact-pauli-decomposition}, we can expand the Hamiltonian terms in $H$ as follows:\footnote{Here, the subscript indicates the qubit and we omit single-qubit identity operators for ease of notation.}
\begin{align}
h_{1,2} &= \sum_{\alpha=0}^3 A^{(\alpha)}_1 \otimes \sigma_2^{(\alpha)}  \label{eq:pauli-decomp-h12}\\
h_{1,3} &= \sum_{\beta=0}^3 B^{(\beta)}_1 \otimes \sigma_3^{(\beta)}  \label{eq:pauli-decomp-h13} 
\end{align}
for some collections of Hermitian matrices $\{A_1^{(\alpha)}\}_\alpha$ and $\{B_1^{(\beta)}\}_\beta$ in $\mathbb{C}^{2 \times 2}$.
Since we know that $\|[h_{1,2},h_{1,3}]\| \leq \eps$, it is tempting to argue that the almost-commuting property must also "propagate" down to the collections of matrices $\{A^{(\alpha)}\}_\alpha$ and $\{B^{(\beta)}\}_\beta$. In \Cref{thm:comm-propagates-down}, we show that this is indeed the case, leveraging the fact that $\{\sigma^{(i)}\}_{i=0}^3$ form an orthogonal basis of the space of $2\times 2$ complex matrices under the Hilbert-Schmidt inner product. Concretely, we can argue that
\begin{align}\label{eq:almost-commuting-A-B}
\|[A_1^{(\alpha)}, B_1^{(\beta)}]\| \leq \eps, \quad \forall 
\alpha,\beta=0,1,2,3.
\end{align}
This observation allows us to effectively reduce the task of rounding a $2$-qubit operator to that of rounding a collection of single-qubit operators; in other words, the goal now is to find \emph{nearby} matrices $\{\hat A_1^{(\alpha)}\}_\alpha$ and $\{\hat B_1^{(\beta)}\}_\beta$ which perfectly pairwise commute i.e. $[\hat A_1^{(\alpha)}, \hat B_1^{(\beta)}]=0$, for all $\alpha,\beta$. 

\paragraph{Pinching.} How can we make two almost-commuting $2 \times 2$ Hermitian matrices commute? Here, we take inspiration from the well-known \emph{pinching technique} for projectors (e.g.~\cite{overlapping}). Suppose that $A,B \in \mathbb{C}^{2 \times 2}$ are arbitrary Hermitian matrices such that $\|[A,B]\| \leq \eps$ and $\Delta(A)>0$; for example, as in \eqref{eq:almost-commuting-A-B} when $A$ is not scalar. First, we observe that the spectral theorem allows us to expand the matrix $A$ as 
\begin{align}
A = \lambda_{\min}(A) \, \Pi + \lambda_{\max}(A) \, (I -\Pi),
\end{align}
where $\lambda_{\min},\lambda_{\max} \in \mathbb{R}$ are eigenvalues and $\Pi$ and $(I -\Pi)$ are orthogonal projectors onto the eigenspaces. To \emph{pinch} $B$ by $A$ (concretely, by the eigenspace projectors associated with $A$), we map
\begin{align}\label{eq:ex-pinchining}
B \quad \mapsto \quad \hat{B} = \Pi \, B \, \Pi + (I-\Pi) B (I-\Pi).     
\end{align}
Because $\hat{B}$ is now diagonal in the eigenbasis of $A$, we can see that $[\hat{B},A]=0$, and hence the two matrices perfectly commute. Moreover, a simple argument (formally proven in \Cref{thm:pinch}) reveals that the distance between $\hat B$ and $B$ is at most
\begin{align}\label{eq:pinching-for-B}
\|\hat{B} - B\| \, \leq \, \frac{\|[A,B]\|}{\Delta(A)} \, \leq \, \frac{\eps}{\Delta(A)} \, ,
\end{align}
where $\Delta(A) = |\lambda_{\max}(A) - \lambda_{\min}(A)|$ is the spectral gap of $A$. This immediately presents a problem: if the eigenvalues of $A$ are extremely close to one another, the spectral gap $\Delta(A)$ can be vanishingly small, thereby causing the distance between $\hat{B}$ and $B$ to blow up. In general, we may not have any control over the spectral properties that arise from the Pauli decompositions in \eqref{eq:pauli-decomp-h12} and \eqref{eq:pauli-decomp-h13}.

\paragraph{Gap or snap?} In order to explain how we get around the spectral gap issue, let us again revisit the example of two Hermitian matrices $A,B \in \mathbb{C}^{2 \times 2}$ with $\|[A,B]\| \leq \eps$. We introduce a \emph{threshold parameter} $\eta \in (0,1)$ (to be determined later) and distinguish between the following cases:
\begin{itemize}
    \item (the \emph{gapped} case) $\Delta(A) \geq \eta$: in this case, the spectrum of $A$ is sufficiently gapped, and the pinching bound in \eqref{eq:pinching-for-B} translates into a stable upper bound of $\|\hat{B} - B\| \leq  \, \frac{\eps}{\eta}$.

    \item (the \emph{nearly-degenerate} case) $\Delta(A) < \eta$: in this case, the pinching bound in \eqref{eq:pinching-for-B} may blow up; fortunately, however, a small gap $\Delta(A)$ necessarily means that $A$ must also be \emph{close} to a multiple of the identity. Instead of pinching $B$, we can simply \emph{snap} $A$ towards $\hat{A} = \frac{\mathrm{Tr}(A)}{2} I$. Importantly, this makes $\hat{A}$ trivially commute with not just $B$ but in fact any other matrix as well---an important fact we will crucially exploit later on.
    A simple calculation (formally shown in \Cref{lemma:snapping-hermitian-qubit-hermitian-operator}) reveals that $\|\hat{A} - A\| = \Delta(A)/2$, and thus $\|\hat{A} - A\| < \eta$.
\end{itemize}
For example, letting $\eta = \sqrt{\eps}$, we can always find a \emph{nearby} commuting matrix within distance $\sqrt{\eps}$; either by pinching $B$ to get $\hat B$ such that $[\hat{B},A]=0$ and $\|\hat{B} - B\| \leq \frac{\eps}{\eta} = \sqrt{\eps}$, or by snapping $A$ to get $\hat{A}$ such that $[B,\hat{A}]=0$ and $\|\hat{A} - A\| < \eta = \sqrt{\eps}$.
The $2 \times 2$ structure is crucial here. In higher local dimension, a non-scalar operator can have a degenerate eigenspace while remaining far from every scalar matrix; for example, $\mathrm{diag}(0,0,1)$ is at operator-norm distance $1/2$ from the nearest scalar matrix. Pinching by such an operator leaves unrestricted matrices within its two-dimensional eigenspace, so two pinched operators need not commute. Thus, the qubit argument does not directly extend to qudits.

\paragraph{Locality-preserving rounding.} 

Equipped with the "gap or snap" paradigm, we explain a slightly simplified version of our full \emph{locality-preserving rounding technique} on the $\eps$-almost commuting Hamiltonian $H = h_{1,2} + h_{1,3} + h_{2,3}$ from before. (We will indicate the step that was simplified when we reach it in our explanation.) 
To this end, we once again use local expansions in terms of the Pauli basis $\{\sigma^{(i)}\}_{i=0}^3$. While our earlier expansions in \eqref{eq:pauli-decomp-h12} and \eqref{eq:pauli-decomp-h13} only involved one of the qubit systems, our full rounding technique makes use of Pauli expansions on both of the qubits; that is, for each of the terms $h_{12}, h_{13}$ and $h_{23}$, we consider both possible expansions as follows:
\begin{align}
h_{1,2} &= \sum_{\alpha=0}^3 A^{(\alpha)}_1 \otimes \sigma_2^{(\alpha)} = \sum_{\alpha=0}^3 \sigma^{(\alpha)}_1 \otimes A_2^{(\alpha)}  \quad\quad\quad \text { for } \quad \{A_1^{(\alpha)}\}_{\alpha=0}^3\, , \quad \{A_2^{(\alpha)}\}_{\alpha=0}^3\quad\label{eq:pauli-decomp-h12-full}\\
h_{1,3} &= \sum_{\beta=0}^3 B^{(\beta)}_1 \otimes \sigma_3^{(\beta)} = \sum_{\beta=0}^3 \sigma^{(\beta)}_1 \otimes B_3^{(\beta)}   \quad\quad\quad \text { for } \quad \{B_1^{(\beta)}\}_{\beta=0}^3\, , \quad \{B_3^{(\beta)}\}_{\beta=0}^3\quad\label{eq:pauli-decomp-h13-full} 
\\h_{2,3} &= \sum_{\gamma=0}^3 C^{(\gamma)}_2 \otimes \sigma_3^{(\gamma)} = \sum_{\gamma=0}^3 \sigma^{(\gamma)}_2 \otimes C_3^{(\gamma)}   \quad\quad\,
\,\,\,\,\,\,\,\text { for } \quad \{C_2^{(\gamma)}\}_{\gamma=0}^3\, , \quad \{C_3^{(\gamma)}\}_{\gamma=0}^3 \quad\label{eq:pauli-decomp-h23-full}
\end{align}
For convenience' sake, we refer to the left decomposition in \Cref{eq:pauli-decomp-h12-full} as the decomposition "about qubit $1$", and the right decomposition in the same equation as the decomposition "about qubit $2$"---and similarly for the others. Note that each expansion introduces a collection of four Hermitian matrices in $\mathbb{C}^{2 \times 2}$: in the decomposition "about qubit $i$", these four matrices act on qubit $i$. In a slight abuse of notation, we use the same variable and rely on the subscript to indicate which qubit the matrices act on, as well as which of the two sets of matrices we are referring to. Thanks to our previous insight that the $\eps$-almost-commuting property of $h_{1,2}, h_{1,3}$ and $h_{2,3}$ "propagates" down to the collections of matrices acting on the same qubit (formally, \Cref{thm:comm-propagates-down}), we get
\begin{align}
\|[A_1^{(\alpha)}, B_1^{(\beta)}]\| &\leq \eps, \quad\quad \forall 
\alpha,\beta=0,1,2,3 \label{eq:eps-almost-comm-1}\\
\|[A_2^{(\alpha)}, C_2^{(\gamma)}]\| &\leq \eps, \quad\quad \forall 
\alpha,\gamma=0,1,2,3 \label{eq:eps-almost-comm-2}\\
\|[B_3^{(\beta)}, C_3^{(\gamma)}]\| &\leq \eps, \quad\hspace{3.7mm} \forall 
\beta,\gamma=0,1,2,3\label{eq:eps-almost-comm-3}
\end{align}
We may visualize the commutation structure of the triangle Hamiltonian as a graph $G=(V,E)$; here the vertex set $V$ consists of the $24$ matrices in \eqref{eq:pauli-decomp-h12-full}, \eqref{eq:pauli-decomp-h13-full} \eqref{eq:pauli-decomp-h23-full}, and $E$ contains edges between vertex pairs which are $\eps$-almost commuting according to~\eqref{eq:eps-almost-comm-1}, \eqref{eq:eps-almost-comm-2} and \eqref{eq:eps-almost-comm-3}. We can group vertices belonging to the same Hamiltonian term together, and we can use a top and bottom vertex layer within each cluster to represent the two possible Pauli expansions (e.g., see~\Cref{fig:triangle-Hamiltinonian-rounding}).

\begin{figure}[t]
    \centering
\begin{tikzpicture}[
    vertex/.style={circle, draw, fill=white, minimum size=8pt, inner sep=0pt},
    vgreen/.style={vertex, fill=ForestGreen!35},
    vred/.style={vertex, fill=BrickRed!60},
    shadow/.style={circle, draw, fill=white, minimum size=8pt, inner sep=0pt},
    edge/.style={draw, gray!80, line width=0.5pt},
    scale=1.2
]

\usetikzlibrary{backgrounds}
\pgfdeclarelayer{bg}
\pgfsetlayers{bg,main}

\def\layersep{0.34}

\def\R{2.3}

\def\d{0.6}

\def\A{90}
\def\B{210}
\def\C{330}

\def\At{0}
\def\Bt{120}
\def\Ct{240}

\foreach \i [count=\k from 1] in {-1.5,-0.5,0.5,1.5} {

    \begin{pgfonlayer}{bg}
        \node[shadow] (As\k)
            at ({(\R+\layersep)*cos(\A) + \i*\d*cos(\At)},
                {(\R+\layersep)*sin(\A) + \i*\d*sin(\At)}) {};
    \end{pgfonlayer}

    \node[vgreen] (A\k)
        at ({\R*cos(\A) + \i*\d*cos(\At)},
            {\R*sin(\A) + \i*\d*sin(\At)}) {};
}

\foreach \i [count=\k from 1] in {-1.5,-0.5,0.5} {

    \begin{pgfonlayer}{bg}
        \node[shadow] (Bs\k)
            at ({(\R+\layersep)*cos(\B) + \i*\d*cos(\Bt)},
                {(\R+\layersep)*sin(\B) + \i*\d*sin(\Bt)}) {};
    \end{pgfonlayer}

    \node[vgreen] (B\k)
        at ({\R*cos(\B) + \i*\d*cos(\Bt)},
            {\R*sin(\B) + \i*\d*sin(\Bt)}) {};
}

\begin{pgfonlayer}{bg}
    \node[shadow] (Bs4)
        at ({(\R+\layersep)*cos(\B) + 1.5*\d*cos(\Bt)},
            {(\R+\layersep)*sin(\B) + 1.5*\d*sin(\Bt)}) {};
\end{pgfonlayer}

\node[vred] (B4)
    at ({\R*cos(\B) + 1.5*\d*cos(\Bt)},
        {\R*sin(\B) + 1.5*\d*sin(\Bt)}) {};

\foreach \i [count=\k from 1] in {-1.5,-0.5,0.5,1.5} {

    \begin{pgfonlayer}{bg}
        \node[shadow] (Cs\k)
            at ({(\R+\layersep)*cos(\C) + \i*\d*cos(\Ct)},
                {(\R+\layersep)*sin(\C) + \i*\d*sin(\Ct)}) {};
    \end{pgfonlayer}

    \node[vertex] (C\k)
        at ({\R*cos(\C) + \i*\d*cos(\Ct)},
            {\R*sin(\C) + \i*\d*sin(\Ct)}) {};
}

\foreach \i in {1,...,4} {
    \foreach \j in {1,...,3} {
        \draw[draw=ForestGreen!90, line width=0.5pt] (A\i) -- (B\j);
    }
    \foreach \j in {1,...,4} {
        \draw[edge] (As\i) -- (Cs\j);
        \draw[edge] (Bs\i) -- (C\j);
    }
}

\foreach \i in {1,...,4} {
    \draw[draw=ForestGreen!90, line width=0.5pt] (A\i) -- (B4);
}

\node at ({(\R+0.9)*cos(\A)}, {(\R+0.9)*sin(\A)}) {\small $h_{1,3}$};
\node at ({(\R+0.9)*cos(\B)}, {(\R+0.9)*sin(\B)}) {\small $h_{1,2}$};
\node at ({(\R+0.9)*cos(\C)}, {(\R+0.9)*sin(\C)}) {\small $h_{2,3}$};

\end{tikzpicture}

    \caption{\textbf{The commutation graph of the $3$-qubit triangle Hamiltonian.} 
    The vertices represent matrices from the local Pauli expansion of each Hamiltonian term;
the inner and outer rows show the expansions about qubits $1,3$ for $h_{1,3}$, $1,2$ for $h_{1,2}$, and $2,3$ for $h_{2,3}$, respectively. The edges represent near-commutation 
across Hamiltonian terms, as in~\eqref{eq:eps-almost-comm-1}, 
\eqref{eq:eps-almost-comm-2} and \eqref{eq:eps-almost-comm-3}. For each 
qubit, our rounding procedure selects a \emph{pivot} with a large 
spectral gap (e.g., the red vertex) and uses it to \emph{pinch} all of 
the relevant vertices belonging to the same qubit (e.g. the green 
vertices). While vertices belonging to the same cluster are not 
required to 
commute, we also choose to \emph{pinch} within clusters to exploit 
\emph{transitivity}: if two green vertices each commute with the gapped 
pivot, they must also commute between each other. Pinching on the qubit 
shared between $h_{1,2}$ and $h_{1,3}$, the \emph{pinching operator} 
$\mathcal{P}_1$ ensures that the pinched vertices all mutually commute 
(now, green edges); this implies that 
 pinched Hamiltonian terms now exactly commute as $[\mathcal{P}_1(h_{1,2}),\mathcal{P}_1(h_{1,3})]=0$.
 }\label{fig:triangle-Hamiltinonian-rounding}
\end{figure}

Choose $\eta =\eps^{1/3}$. Then, our locality-preserving rounding takes place as follows:
\begin{enumerate}
    \item (Partitioning:) Split the almost-commuting Hamiltonian $H$ into a \emph{gapped} and \emph{nearly-degenerate} part (depending on $\eta$) such that $H= H_{\gap} + H_{\deg}$, where
    \begin{itemize}
        \item $H_{\gap}$ features Hamiltonian terms for which both sets of the two possible Pauli expansions contain a matrix of spectral gap of at least $\eta$ (these particular matrices will enable us to round to nearby commuting matrices via \emph{pinching}), and

        \item $H_{\deg}$ features Hamiltonian terms that do not satisfy the gap condition: morally, these terms act "almost trivially" on at least one of the qubits they touch. 
    \end{itemize}
    \item (Snap:) For each term $h_{i,j} \in H_{\deg}$, we use snapping to round it to a term that is \emph{exactly} trivial on at least one qubit. Concretely, suppose $h_{1,2}$ is in $H_{\deg}$ because its decomposition
    \[ h_{1,2} = \sum_{\alpha=0}^{3} A_1^\alpha \otimes \sigma_2^{(\alpha)}\]
    is not gapped---that is, each matrix $A_1^\alpha$ has spectral gap smaller than $\eta$. Then snapping simply consists of replacing each $A_1^\alpha$ in the decomposition with $\frac{\mathrm{Tr}(A_1^\alpha)}{2}I_1$, yielding a new operator $h'_{1,2}$ that is effectively 1-local, and may in fact be simply a multiple of the identity. As a \textbf{simplifying assumption}, we will suppose that all the post-snapping operators $h'_{i,j}$ are in fact multiples of identity: the general case including 1-local operators is more cumbersome, requiring an additional round of snapping with the same gap parameter, but conceptually no more difficult than this special case. Let $H'$ be $H$ after snapping has been performed (with all the terms in $H_{\gap}$ unchanged). This step incurs an error of $O(\eta)$.

    \item (Gap:) We are now in the position that for every Hamiltonian term that acts nontrivially on a qubit, at least one operator in its Pauli decomposition for that qubit is guaranteed to have a spectral gap. Intuitively, the spectral gap means these gapped operators have a "strong opinion" about the correct basis in which to pinch this qubit, and our algorithm will listen to them. More precisely, for any qubit $i$ that is nontrivially acted on by at least two terms of $H'$, pick \emph{any} gapped operator in the decomposition "about qubit $i$" of the terms acting on $i$, and set it to be the \emph{pivot} $R_i$ associated with $i$. (If a qubit does not have two nontrivial terms, set the pivot to be identity.) Once a pivot has been identified for each qubit, \emph{pinch} each Hamiltonian term $h'_{i,j}$ simultaneously by the pivots $R_i$ and $R_j$: mathematically,
    \[ h'_{i,j} \mapsto \hat{h}_{i,j} = (\cP_i \otimes \cP_j)(h'_{i,j}),\]
    where $\cP_i$ is the \emph{pinching superoperator} given by
 \begin{align}\label{eq:pinching-superoperator-overview}
    \mathcal{P}_i(X) = \Pi_i \, X \, \Pi_i + (I-\Pi_i) X (I-\Pi_i)\, 
    \end{align}
and where $\Pi_i$ and $I-\Pi_i$ are the eigenspace projectors of the $i$-th pivot element; when $R_i=I$, we set $\mathcal{P}_i=\mathrm{id}$. The Hamiltonian $\hat{H} = \sum_{(i,j)} \hat{h}_{i,j}$ is the final product of our rounding procedure.

After the pinching process is complete, it is easy to see that all the terms $\hat{h}_{i,j}$ are forced to commute with each other by construction. What is harder to see is that this does not incur a large error. This is essentially for two reasons:
\begin{itemize}
    \item \textbf{Propagation:} If we imagine pinching $h'_{i,j}$ by a pivot arising from a \emph{different} Hamiltonian term $h'_{i,k}$, then this pivot is guaranteed to almost commute---up to error $\eps$---with all terms in the decomposition of $h'_{i,j}$ by the fact that $h'_{i,j}$ and $h'_{i,k}$ approximately commute, together with \Cref{thm:comm-propagates-down}.
    \item \textbf{Transitivity:} If we imagine pinching $h'_{i,j}$ by a pivot arising from the \emph{same} Hamiltonian term $h'_{i,j}$, then this pivot can be shown to almost commute with all other terms in the decomposition of $h'_{i,j}$ by using the \emph{transitivity} of commutation for gapped matrices, by passing through another term $h'_{i,k}$. This is illustrated in \Cref{fig:triangle-Hamiltinonian-rounding}. The transitivity uses the gapped property of $h'_{i,k}$, and yields a slightly worse commutator error of $O(\eps/\eta)$.
\end{itemize}
Thus, the pivot with which we pinch a given Hamiltonian term commutes with that term up to error at most $\kappa = O(\eps/\eta)$, and it can be shown that the error induced in pinching is at most $O(\kappa / \eta) = O(\eps/\eta^2)$. (We refer to \Cref{lemma:double-extra-cheese} for a formal analysis of the pinching operation.) Our choice of $\eta = \eps^{1/3}$ balances the errors from the gap and snap steps, giving us an overall error of $O(\eps^{1/3})$. In the more general case, where our simplifying assumption in the snapping step is removed, it turns out a second round of snapping is needed: since snapping does not increase pairwise commutator norms, this second round uses the same gap parameter $\eta = \eps^{1/3}$ and contributes $O(\eps^{1/3})$ error, leaving the final bound $O(\eps^{1/3})$ unchanged.
\end{enumerate}

Notice that the Hamiltonian $\hat H$ inherits the same locality from the original Hamiltonian $H$; this is due to the \emph{local} nature of our pinching and snapping techniques.
In \Cref{thm:transforming-hamiltonian}, we show that our rounding technique---when  applied to the triangle Hamiltonian---yields an \emph{exactly commuting} Hamiltonian $\hat{H} = \hat h_{1,2} + \hat h_{2,3} + \hat h_{1,3}$ with $[\hat h_{1,2},\hat h_{1,3}] = [\hat h_{1,3},\hat h_{2,3}] = [\hat h_{1,2},\hat h_{2,3}] = 0$ and
$$
\| H - \hat{H} \| \, \leq \, O(\eps^{1/3}).
$$
Our rounding technique thus rests on two important structural facts: (i) approximate commutation of
Hamiltonian terms propagates to the single-qubit operators in their Pauli
decompositions, and (ii) \emph{pinching} against gapped local pivots enforces exact
commutation while preserving locality.  In this sense, our procedure is a
consistent local change of basis across the interaction graph, with \emph{snapping}
handling terms which are too close to the identity to select a stable basis. 

\paragraph{Concrete $3$-qubit example.}

We now illustrate the rounding procedure on a concrete triangle Hamiltonian
on three qubits, shown in~\Cref{fig:triangle2}. 

The terms $h_{1,2}$ and $h_{2,3}$ commute exactly: both are controlled on
the shared qubit, qubit $2$, in the standard basis.  On the other hand,
$h_{1,3}$ acts in the $X$ basis on qubits $1$ and $3$, and therefore
commutes only approximately with $h_{1,2}$ and $h_{2,3}$; the obstruction
comes from the $Z$-basis components of $h_{1,2}$ and $h_{2,3}$, whose
coefficients are suppressed by the factor $1/100$.

This example is useful because it rules out a na\"{i}ve sequential
approach to rounding.  One might try first to round the already commuting
pair $h_{1,2},h_{2,3}$ to an exactly commuting pair
$\hat h_{1,2},\hat h_{2,3}$, without taking $h_{1,3}$ into account, and only
afterwards try to round $h_{1,3}$ to some $\hat h_{1,3}$ that commutes with
both of them.  However, for the particular pair $h_{1,2},h_{2,3}$ above,
there is no operator on qubits $1$ and $3$ that acts nontrivially on both
qubits and commutes exactly with both $h_{1,2}$ and $h_{2,3}$.  Thus any
successful rounding must modify at least one of $h_{1,2}$ or $h_{2,3}$,
despite the fact that these two terms commute perfectly with each other.
This illustrates why the rounding procedure must be global in nature: the
choice of how to round one pair of terms cannot be made independently of
the remaining interactions.



\begin{figure}[t]
\centering
\begin{tikzpicture}[
    scale=1.25,
    interaction/.style={draw=black!75, line width=1.1pt},
    qubit/.style={
        circle,
        draw=black!70,
        fill=gray!10,
        line width=0.7pt,
        minimum size=4.8mm,
        inner sep=0pt,
        font=\scriptsize
    },
    edgelabel/.style={
        font=\small,
        inner sep=1pt
    },
    equations/.style={
        align=left,
        font=\small
    }
]

\coordinate (P1) at (0,0);
\coordinate (P2) at (2.2,0);
\coordinate (P3) at (1.1,1.9);

\draw[interaction] (P1) -- (P2) -- (P3) -- cycle;

\node[qubit] at (P1) {$1$};
\node[qubit] at (P2) {$2$};
\node[qubit] at (P3) {$3$};

\node[edgelabel] at ($(P1)!0.5!(P2) + (0,-0.22)$) {$h_{1,2}$};
\node[edgelabel] at ($(P2)!0.5!(P3) + (0.25,0.13)$) {$h_{2,3}$};
\node[edgelabel] at ($(P3)!0.5!(P1) + (-0.25,0.13)$) {$h_{1,3}$};

\node[equations, anchor=west] at (3.35,0.95) {%
$\displaystyle
\begin{aligned}
h_{1,2} &= X \otimes \proj{0} \otimes I
  + \frac{1}{100}\, Z \otimes \proj{1} \otimes I, \\[1.5mm]
h_{2,3} &= I \otimes \proj{0} \otimes X
  + \frac{1}{100}\, I \otimes \proj{1} \otimes Z, \\[1.5mm]
h_{1,3} &= X \otimes I \otimes X .
\end{aligned}
$
};

\end{tikzpicture}
\caption{A simple triangle Hamiltonian on three qubits.}
\label{fig:triangle2}
\end{figure}

Let us now see how the one-shot rounding procedure handles this example.
All three terms are gapped in the relevant sense in both of their possible
Pauli decompositions, so the partitioning and snapping steps are trivial.
The only nontrivial part of the procedure is the gapped rounding step, in
which we choose a pivot operator $A_i$ for each qubit $i$ and then pinch
the Hamiltonian terms with respect to these pivots.

We begin with qubit $2$.  The terms acting on qubit $2$ are already
written in the appropriate form, with the Pauli operators placed on the
non-shared qubits $1$ and $3$:
\begin{align*}
    h_{1,\blue{2}}
        &= \blue{(\proj{0})_2} \otimes X_1
        + \frac{1}{100} \blue{(\proj{1})_2} \otimes Z_1, \\
    h_{\blue{2},3}
        &= \blue{(\proj{0})_2} \otimes X_3
        + \frac{1}{100} \blue{(\proj{1})_2} \otimes Z_3.
\end{align*}
The coefficient $\proj{0}$ is gapped, so we take
\[
    A_2 = (\proj{0})_2 .
\]

Next consider qubit $1$.  The term $h_{1,3}$ is already expressed in the
desired form, while $h_{1,2}$ must be rewritten so that the Pauli operators
act on the non-shared qubit $2$:
\begin{align*}
    h_{\red{1},2}
        &= \frac{1}{2} \red{\left(X + \frac{1}{100} Z\right)_1}
            \otimes I_2
        + \frac{1}{2} \red{\left(X - \frac{1}{100} Z\right)_1}
            \otimes Z_2, \\
    h_{\red{1},3}
        &= \red{X_1} \otimes X_3 .
\end{align*}
There are several possible gapped choices for the pivot.  For simplicity,
we take
\[
    A_1 = X_1 .
\]

Finally, for qubit $3$, we proceed in the same way:
\begin{align*}
    h_{2,\DarkGreen{3}}
        &= \frac{1}{2} I_2 \otimes
            \DarkGreen{\left(X + \frac{1}{100} Z\right)_3}
        + \frac{1}{2} Z_2 \otimes
            \DarkGreen{\left(X - \frac{1}{100} Z\right)_3}, \\
    h_{1,\DarkGreen{3}}
        &= X_1 \otimes \DarkGreen{X_3}.
\end{align*}
Again, we choose the convenient pivot
\[
    A_3 = X_3 .
\]

With these choices of pivots, the pinching step is especially transparent:
it forces the Hamiltonian to be diagonal in the tensor-product basis
defined by the $X$ basis on qubit $1$, the standard basis on qubit $2$,
and the $X$ basis on qubit $3$.  The rounded Hamiltonian
\[
    \hat H = \hat h_{1,2} + \hat h_{1,3} + \hat h_{2,3}
\]
is therefore given by
\begin{align*}
    \hat h_{1,2} &= X_1 \otimes \proj{0}_2, \\
    \hat h_{1,3} &= X_1 \otimes X_3, \\
    \hat h_{2,3} &= \proj{0}_2 \otimes X_3.
\end{align*}
These three terms now commute exactly.

\paragraph{The role of transitivity.}

This simple example also illustrates the role played by \emph{transitivity} of
commutation for $2 \times 2$ matrices.  As discussed above, transitivity
is what allows the rounding procedure to control commutation relations
among different operators appearing in the Pauli decomposition of a
single Hamiltonian term.  In the present example, this control is reflected
in the fact that the small $Z$-components in $h_{1,2}$ and $h_{2,3}$ are
compatible with the choice of $X$-basis pivots on qubits $1$ and $3$.

One way to see this is to ask what would happen if the coefficient
$1/100$ in $h_{1,2}$ were replaced by $1$.  Then $h_{1,2}$ and $h_{2,3}$
would still commute exactly, since they would remain diagonal in the same
basis on the shared qubit $2$.  However, the commutator between $h_{1,2}$
and $h_{1,3}$ would become large.  In fact, $h_{1,2}$ would no longer
commute with any operator on qubits $1$ and $3$ that acts nontrivially on
both qubits.  Thus the smallness of the noncommuting component is
essential: it is precisely what permits the global choice of pivots to
round all three terms simultaneously.

\paragraph{Generalizing to $m$ Hamiltonian terms.}

The same mechanism extends from this triangle example to general
$2$-local qubit Hamiltonians with $m$ terms.  Rather than choosing a
rounding separately for each pair of Hamiltonian terms, the algorithm
chooses pivots at the level of individual qubits.  Each Hamiltonian term
is then pinched with respect to the pivots on the qubits it touches.  This
one-shot choice is what preserves locality while enforcing exact
commutation globally: at each qubit shared by two terms, their local components are diagonal in the same pivot basis, so every pair of rounded terms now commutes exactly.

The example above captures the main reason such a global procedure is
needed.  Even when two terms commute exactly, keeping them fixed may be
incompatible with rounding the rest of the Hamiltonian.  The rounding
algorithm therefore treats the Hamiltonian as a whole, using the
approximately commuting structure to identify local bases that are
simultaneously compatible across all interactions.

\subsection{Discussion and Future Work}
\label{sec:future}
The most immediate next question to investigate is whether the $\eps$-dependence of our approximation result can be improved to match our lower bound of $\Omega(\sqrt{\eps})$. Preliminary AI-assisted investigations suggest a possible refinement of our pinching and snapping techniques with an $O(\sqrt{\eps})$ rounding error. This suggestion remains unverified, is not a result of the present paper, and requires a complete proof.

After improving the $\eps$ dependence, the next natural questions to investigate are whether our rounding can be extended to systems of higher locality or local dimension. We suspect that going beyond qubits may be challenging, since the transitivity property that is crucial to our proof is no longer true for higher dimensions (even for dimension 3). 

Another direction for generalization is to study looser notions of approximate commutation. For instance, what can we say if we only have a bound on the \emph{average} norm of the commutators? Similar regimes were studied algorithmically in~\cite{schmidhuber2025hamiltoniandecodedquantuminterferometry}. 

In the context of complexity theory, one interesting direction is to see what our result implies about the quantum PCP conjecture. Recall that our rounding result already implies that any qPCP-hard family of 2-local qubit Hamiltonians must exhibit some moderate amount of noncommutativity: assuming $\mathsf{NP} \neq \mathsf{QMA}$, any
family of $2$-local qubit Hamiltonians for which the $\gamma$-LH problem is $\QMA$-hard for \emph{constant} $\gamma$ must allow for at least one of the pairwise commutator norms to be $\Omega(1)$. (In fact, such a family must also allow for a constant fraction of the terms to fail to approximately commute with the other terms---otherwise, these ``bad terms'' could simply be deleted from the Hamiltonian.) It would be interesting to see whether this bound can be improved, for instance by finding a rounding technique that is more sensitive to the commutator norms of individual pairs of terms, rather than merely the maximum commutator norm.

Taking an alternate perspective, the complexity of the \emph{general} $k$-local commuting Hamiltonian problem on qudits is wide open, and it is possible that some version of this problem is even qPCP-hard. Moreover, it is possible that gap-amplifying transformations for commuting Hamiltonians will be easier to design than for general Hamiltonians. Could some rounding scheme along the lines of ours show that gap amplification for commuting Hamiltonians directly implies the quantum PCP conjecture?

\paragraph{Acknowledgements.}
Special thanks to Thomas Vidick for an inspirational email conversation at the start of this project. We also thank Christopher Laumann, Jiaqing Jiang, Thomas Vidick, and Rahul Trivedi for useful conversations on the general status of commuting and almost commuting Hamiltonians and the role they play in physics. 

IF thanks Muhammad Faisal for help with setting up an effective environment to use generative AI. AN was partially supported by NSF CAREER Award 2339948.

\paragraph{Author contributions and use of AI.}
All authors contributed to this work and share responsibility for its content.
The majority of this work was completed between September 2025 and February 2026 without substantial AI usage.
We used ChatGPT for support with our literature review and the design of our figures (from September 2025 - February 2026), for pointing out a small modification to our proof to improve the error bound from $\varepsilon^{1/6}$ to $\varepsilon^{1/3}$ (in September 2026), for providing the example in \cref{sec:lower-bound} (in September 2026), and for help with final editing (in September 2026). AI tools were also used for the preliminary investigations discussed in \cref{sec:future}. In November 2026, ChatGPT assisted with publication-stage textual revisions, LaTeX formatting, and reference checks. Apart from these disclosed uses, the technical results and manuscript text were developed and written by the authors. The authors are responsible for verifying the AI-assisted material and for the accuracy and integrity of the final manuscript.

\section{Preliminaries}
\label{sec:background}

The $d$-dimensional complex vector space is denoted by $\mathbb{C}^d$. We use $\mathrm{L}(V,W)$ to denote the set of linear operators mapping between finite-dimensional complex vector spaces $V$ and $W$, and we use
$\mathrm{L}(\mathbb{C}^d)$ to denote the set of linear operators over $\mathbb{C}^d$. The \emph{operator norm} of a linear operator $A \in \mathrm{L}(V,W)$, for some finite-dimensional complex vector spaces $V \neq \{0\}$ and $W$ is given by
\begin{align}
\|A \| = \sup_{\|x\|_2 =1} \|A x\|_2 \, ,    
\end{align}
where $\|x \|_2 = \sqrt{\langle x,x \rangle}$ is the standard Euclidean norm; we occasionally omit the subscript for convenience. The operator norm satisfies the triangle inequality and is sub-multiplicative with $\|A \cdot B\| \leq \|A\| \cdot \|B\|$, for all $A \in \mathrm{L}(V,W)$ and $B \in \mathrm{L}(U,V)$, where $U$ is a finite-dimensional complex vector space.
A linear operator $A \in \mathrm{L}(\mathbb{C}^d)$ is called \emph{positive semi-definite} (PSD), denoted by $A \succeq 0$, whenever $x^\dag A x \geq 0$, for all $x \in \mathbb{C}^d$, and where $\dag$ is the \emph{adjoint}.
A \emph{unitary} $U: \mathbb{C}^d \to \mathbb{C}^d$ is a linear operator such that $U^\dagger U = U U^\dagger = I$, where $I$ is the identity matrix on $\mathbb{C}^d$. 
A linear operator $H \in \mathrm{L}(\mathbb{C}^d)$ is called \emph{Hermitian} if $H^\dag = H$. We use $\mathrm{Herm}(\mathbb{C}^{d})$ to denote the set of $d$-dimensional Hermitian matrices. 
A Hermitian operator $\Pi \in \mathrm{Herm}(\mathbb{C}^{d})$ is called a \emph{projector} if $\Pi^2=\Pi$. Any $A \in \mathrm{Herm}(\mathbb{C}^{2})$ admits a spectral decomposition of the form 
\begin{align}
 A = \lambda_1 \Pi + \lambda_2 (I -\Pi)   
\end{align}
where $\lambda_1,\lambda_2 \in \mathbb{R}$ are the eigenvalues of $A$, and where $\Pi$ and $(I-\Pi)$ are the projectors onto the respective eigenspaces of $A$; we call the eigenvalue difference
\begin{align}
\Delta(A) := |\lambda_1 - \lambda_2|   
\end{align}
the \emph{spectral gap} of $A$, and we say that $A$ is
$\eta$-gapped if $\Delta(A) \geq \eta$.

\paragraph{Commutator.} Let $A,B \in \mathbb{C}^{d \times d}$. Then, the \emph{commutator} between $A$ and $B$ is defined as
\begin{equation}
            [A, B] := AB - BA.
        \end{equation}
We say that $A$ \emph{commutes} with $B$ (equivalently, $B$ commutes with $A$) if $[A, B] = 0$ or $AB = BA$. For a specified $\eps \geq 0$, we say that $A$ and $B$ $\eps$-\emph{almost commute} if
\begin{equation}
\|[A,B]\| \leq \eps \, ,
        \end{equation}
where $\|\cdot\|$ is the operator norm.
We use the following elementary fact about the commutator.
\begin{restatable}{fact}{TransitivityLeadingFact}\label{lem:transitivity-of-almost-commutatitivty}
Let $A,B,C \in \mathbb{C}^{d \times d}$ be arbitrary matrices. Then,
\begin{equation}
\| [A,C] \| \leq 2 \cdot \|A-B\| \cdot \|C\|+ \| [B,C]\|.    
\end{equation}
\end{restatable}

\begin{proof}
First, notice that the commutator is linear in the sense that
\begin{align}
[A,C] = [A-B,C] + [B,C].    
\end{align}
Using the triangle inequality and the sub-multiplicativity of the operator norm, we get
\begin{align}
\| [A,C] \| &= \| [A-B,C] + [B,C] \| \\    
&\leq \| [A-B,C]\| + \|[B,C] \| \\
&= \| (A-B)C - C(A-B)\| + \|[B,C] \| \\
&\leq \| (A-B)C\| + \|C(A-B)\| + \|[B,C] \|\\
&\leq 2 \cdot \|A-B\| \cdot \|C\| + \| [B,C]\|.
\end{align}
\end{proof}
\vspace{-0.2cm}We use the following well-known fact about \emph{commuting} Hermitian matrices.

\begin{fact}[Simultaneous Diagonalization]\label{fact:sim-diag}
    Let $A,B \in \mathrm{Herm}(\mathbb{C}^{2})$ be exactly commuting Hermitian matrices. Then, $A$ and $B$ can be \emph{simultaneously diagonalized} in the following sense: there exists a spectral decomposition via common eigenspace \emph{projectors} $\Pi$ and $(I-\Pi)$ such that
    \begin{equation}
        A = a_1 \Pi + a_2 (I -\Pi) \quad\quad \text{and} \quad\quad B = b_1 \Pi + b_2 (I - \Pi) \, 
    \end{equation}
    where $a_1,a_2 \in \mathbb{R}$ and $b_1,b_2 \in \mathbb{R}$ are the eigenvalues of $A$ and $B$, respectively.
\end{fact}

\paragraph{Spectral stability.} We use Weyl's inequality (e.g., see~\cite{tao2010eigenvalues}) which states that the eigenvalue spectra of Hermitian matrices remain stable under perturbations.

\begin{lemma}[Weyl's inequality]\label{lem:spectral-stab}
Let $A,B \in \mathrm{Herm}(\mathbb{C}^{d})$. Then, for all $i \in [d]$,
$$
\left|\lambda_i(A+B) - \lambda_i(A) \right| \, \leq \, \|B\| \, ,
$$
where $\lambda_i(A)$ and $\lambda_i(A+B)$ denote the $i$-th eigenvalues of $A$ and $A+B$, respectively, listed in the same nondecreasing order.
\end{lemma}

The following is a direct consequence of Weyl's inequality in \Cref{lem:spectral-stab}.
\begin{fact}\label{fact:weyl-corollary}
Let $A,B \in \mathrm{Herm}(\mathbb{C}^{d})$ and $\eps \geq 0$. If $\|A - B\| \leq \eps$, then $|\lambda_{\mathrm{min}}(A) -\lambda_{\mathrm{min}}(B)| \leq \eps$.
\end{fact}

\paragraph{Quantum information.} 
A finite-dimensional complex Hilbert space is denoted by $\mathcal H$, and we use subscripts to distinguish between different systems (or registers); for example, we let $\mathcal H_{1}$ be the Hilbert space corresponding to the system labeled as $1$. 
The tensor product of two Hilbert spaces $\mathcal H_{1}$ and $\mathcal H_{2}$ is another Hilbert space which we denote by $\mathcal H_{1,2} = \mathcal H_{1} \otimes \mathcal H_{2}$.  

A quantum system over the $2$-dimensional Hilbert space $\mathcal H \cong \mathbb{C}^2$ is called a \emph{qubit}; in general, when $\mathcal H \cong \mathbb{C}^d$ for $d \geq 2$, we call it a \emph{qudit}. For $n \in \mathbb{N}$, we refer to quantum registers over the Hilbert space $\mathcal H \cong \big(\mathbb{C}^2\big)^{\otimes n}$ as $n$-qubit states. We use the word \emph{quantum state} to refer to both pure states (unit vectors $\ket{\psi} \in \mathcal H$), as well as density matrices $\rho \in \mathrm{D}(\mathcal H)$ with
\begin{align}
\mathrm{D}(\mathcal H) = \{ \rho \in \mathrm{L}(\mathcal{H}) \, : \, \rho \text{ is Hermitian, $\rho \geq 0$ and  } \mathrm{Tr}[\rho] =1\}.    
\end{align}
Given a bipartite operator $\rho_{1,2} \in \mathrm{D}(\mathcal H_{1,2})$, we denote the partial trace over a system, say the second system, by $\mathrm{Tr}_2[\cdot]$, and we call $\rho_1=\mathrm{Tr}_2[\rho_{1,2}]$ the reduced state. We use the following simple fact:
\begin{restatable}{fact}{FactOPNormPTrace} \label{fact:op-norm-ptrace}
    Let $\mathcal H_{1,2} = \mathcal H_{1} \otimes \mathcal H_{2}$ be a bipartite Hilbert space, for two finite-dimensional complex Hilbert spaces $\mathcal H_{1}$ and $\mathcal H_{2}$ of dimension $d_1$ and $d_2$, respectively. Then, it holds that
    \begin{align}
        \|\rho_{1,2}\| \geq \frac{1}{d_2} \, \|\mathrm{Tr}_2(\rho_{1,2})\|, \quad \quad \forall \rho_{1,2} \in \mathrm{D}(\mathcal H_{1,2}).
        \end{align}
\end{restatable}

\begin{proof}
By the variational characterization of the operator norm, we have
\begin{align*}
\|\rho_{1,2}\| 
&= \max_{\tau_{1,2} \in \mathrm{D}(\mathcal H_{1,2})} \mathrm{Tr}\big(\rho_{1,2} \tau_{1,2}\big) \\
&\ge \max_{\sigma_1 \in \mathrm{D}(\mathcal H_1)} \mathrm{Tr}\Big(\rho_{1,2} \big(\sigma_1 \otimes \frac{I_2}{d_2}\big)\Big) \\
&= \max_{\sigma_1 \in \mathrm{D}(\mathcal H_1)} \frac{1}{d_2} \mathrm{Tr}\big( \mathrm{Tr}_2(\rho_{1,2}) \, \sigma_1 \big) \\
&= \frac{1}{d_2} \, \max_{\sigma_1 \in \mathrm{D}(\mathcal H_1)} \mathrm{Tr}\big( \mathrm{Tr}_2(\rho_{1,2}) \, \sigma_1 \big) \\
&= \frac{1}{d_2} \, \|\mathrm{Tr}_2(\rho_{1,2})\|.
\end{align*}
In the third line, we used the standard identity $\mathrm{Tr}\big(\rho_{12} (\sigma_1 \otimes I_2)\big) = \mathrm{Tr}\big(\mathrm{Tr}_2(\rho_{12}) \, \sigma_1\big)$
to reduce the partial trace explicitly. This completes the proof.
\end{proof}

\paragraph{Pauli matrices.} The four $2\times 2$ \emph{Pauli matrices} are denoted as
    \begin{align}
        \sigma^{(0)} &= I = \begin{pmatrix} 1 & 0 \\0 & 1 \end{pmatrix}, \quad\quad\,\,\,\,\sigma^{(1)} = X = \begin{pmatrix} 0 & 1 \\ 1 & 0 \end{pmatrix}\\ \sigma^{(2)} &= Y = \begin{pmatrix} 0 & -i \\ i & 0 \end{pmatrix}, \quad\quad \sigma^{(3)} = Z = \begin{pmatrix} 1 & 0 \\ 0 & -1 \end{pmatrix}.
    \end{align}
The Pauli matrices are Hermitian and unitary, hence involutions: $(\sigma^{(\alpha)})^2=I$. Each can therefore be viewed as an \emph{observable}; for $\alpha=1,2,3$ its eigenvalues are $\pm1$.

\begin{fact}[Properties of the Pauli matrices]\label{fact:pauli-properties}
The Pauli matrices $\{\sigma^{(\alpha)}\}_{\alpha =0}^3$ satisfy
  \begin{align}\label{trace-of-product-of-paulis}
&\mathrm{Tr}[\sigma^{(0)}]=2, \quad\mathrm{Tr}\big[\sigma^{(\alpha)}\big]= 0, \,\, \forall\alpha =1,2,3, \quad \text{ and }\\
&\mathrm{Tr}\left[\sigma^{(\alpha)} \sigma^{(\beta)}\right] =  2 \delta_{\alpha,\beta}, \quad \forall \alpha,\beta =0,1,2,3\, ,
    \end{align}
and form an orthogonal basis of $\mathbb{C}^{2 \times 2}$ with respect to the Hilbert-Schmidt inner product; in other words, any linear operator $A \in \mathbb{C}^{2 \times 2}$ can be written as a complex linear combination of the form
\begin{align}
A = \frac{1}{2}\sum_{\alpha=0}^3  \mathrm{Tr}\left[\sigma^{(\alpha)} A\right] \cdot \sigma^{(\alpha)}.
\end{align}
\end{fact}

\subsection{Local Hamiltonian Problem}

We begin by recalling the promise-problem formulation of the local Hamiltonian problem.  This language is convenient because local Hamiltonian problems naturally come with an energy gap promise: an instance is assumed to be either below a lower threshold or above an upper threshold, with no requirement on what happens in between.

\begin{definition}[Binary Promise Problem]
Let $\{0, 1\}^*$ be the set of all binary strings. $A = \left(A_{\mathrm{YES}}, A_{\mathrm{NO}}\right)$ is called a binary promise problem when $A_{\mathrm{YES}} \subseteq \{0, 1\}^*$, $A_{\mathrm{NO}} \subseteq \{0, 1\}^*$, and $A_{\mathrm{YES}} \cap A_{\mathrm{NO}} = \emptyset$.
\end{definition}

We now specialize this general notion to the $2$-local Hamiltonian problem.  Throughout the paper, we keep track of both the absolute promise gap and the relative promise gap, since the latter is the natural normalization when the Hamiltonian has $m$ local terms.

\begin{definition}[$2$-\textbf{L}ocal \textbf{H}amiltonian Problem \kLH]
\label{defn:k-LH}
The $2$-local Hamiltonian problem notated as \ \kLH\ is a binary promise problem where the input is a classical binary string $x = (H, a, b)$ such that:
\begin{itemize}
    \item $H$ is a $2$-local Hamiltonian $H = \sum\limits_{(i, j) \in E} h_{i, j}$ on a total of $n$ qubits where $E \subseteq [n] \times [n]$ and $|E| = m = \poly(n)$ and each $h_{i,j}$ is a Hermitian matrix with a bounded operator norm $\|h_{i,j}\| \leq 1$ and its entries are specified by $\poly(n)$ bits and $h_{i,j}$ is non-identity on at most $2$ qubits (namely $i$, $j$),
    \item $a$ and $b$ are rational or constant-degree real algebraic numbers with exact $\poly(n)$-bit descriptions such that $a < b$; the gap $\Gamma = b - a$ is called the \textbf{absolute promise gap} and $\gamma = \Gamma / m$ is called the \textbf{relative promise gap}, and we require $\gamma \geq 1/\poly(n)$,
    \item for yes-instances, there exists an $n$-qubit quantum state $\ket{\psi}$ such that $\bra{\psi}H\ket{\psi} \leq a$ (i.e. energy of the state w.r.t. $H$ is at most $a$),
    \item for no-instances, for every $n$-qubit quantum state $\ket{\psi}$, it holds that $\bra{\psi}H\ket{\psi} \geq b$ (i.e. energy of the state w.r.t. $H$ is at least $b$), and
    \item it is promised that any instance will be either a yes or no instance.
\end{itemize}
\end{definition}

The main objects of this paper are restricted versions of the local Hamiltonian problem in which the local terms commute, or nearly commute, with one another.  The exactly commuting case serves as the classical benchmark against which we compare the almost commuting setting.

\begin{definition}[Commuting $2$-Local Hamiltonian Problem]
    The Commuting Local Hamiltonian Problem notated as $\gamma$-$2$-$\mathrm{CLH}$ is defined as the \kLH\ problem in \Cref{defn:k-LH} with an additional restriction on the input that any pair of terms commute; in other words,
    \begin{equation}
    \left[h_{i, j}, h_{k, \ell}\right] = 0, \quad\quad \forall (i, j) \in E, \,\, \forall (k, \ell) \in E.
    \end{equation}
\end{definition}

We also introduce the almost commuting analogue, where exact commutation is relaxed to a uniform operator-norm bound on all pairwise commutators.  This is the promise problem to which our rounding theorem will later be applied.

\begin{definition}[Almost Commuting Local Hamiltonian Problem]
The Almost Commuting $2$-Local Hamiltonian Problem notated as $(\gamma, \eps)$-$2$-$\mathrm{ACLH}$ is defined as the \kLH\ problem in \Cref{defn:k-LH} with an additional restriction that any pair of terms $\eps$-almost-commute:
\begin{align}
\,\,\,\, \|\big[h_{i, j}, h_{k, \ell}\big]\| \leq \eps \, , \quad\quad \forall (i, j) \in E, \,\, \forall (k, \ell) \in E.
\end{align}
\end{definition}

\subsection{Complexity Classes}

We next recall the complexity-theoretic notions used to state our containment and hardness consequences.  Since the local Hamiltonian problem is naturally a promise problem, we use the corresponding promise-problem formulations of $\NP$ and $\QMA$.

\begin{definition}[$\NP$]
   A binary promise problem~\footnote{We technically define a promise version of $\NP$. Historically, containment in $\NP$ was defined when any binary string was either yes or no. That promise version of $\NP$ has been recently widely adopted when talking about promise problems.} $A = \left(A_{\mathrm{YES}}, A_{\mathrm{NO}}\right)$ is contained in $\NP$ if there exist polynomials $p(n), q(n)$ and a deterministic verifier $V$ such that on instance $x \in A_{\mathrm{YES}} \cup A_{\mathrm{NO}}$ with $|x| = n$:
    \begin{enumerate}
        \item \textbf{Completeness:} $x \in A_{\mathrm{YES}} \: \Longrightarrow \: \exists w: |w| = q(n) \text{ and } V(x, w) \text{ outputs } 1 \text{ in } p(n) \text{ time; and}$
        \item \textbf{Soundness:} $ x \in A_{\mathrm{NO}} \: \Longrightarrow \: \forall w: V(x, w) \text{ outputs } 0 \text{ in } p(n) \text{ time.}$
    \end{enumerate}
\end{definition}

In contrast to $\NP$, the class $\QMA$ allows the verifier to receive a quantum witness and to perform a quantum verification procedure.  This is the natural quantum analogue of $\NP$ and is the complexity class captured by the general local Hamiltonian problem.

\begin{definition}[$\QMA$]
   A promise problem $A = \left(A_{\mathrm{YES}}, A_{\mathrm{NO}}\right)$ is contained in $\QMA$ if there exist polynomials $p(n), q(n), r(n)$ and a uniform family of quantum verifiers $\{V_n\}$ (that run in $p(n)$ time) such that on instance $x \in A_{\mathrm{YES}} \cup A_{\mathrm{NO}}$ with $|x| = n$:
    \begin{enumerate}
        \item \textbf{Completeness:} If $x$ is a yes-instance, i.e. $x \in A_{\mathrm{YES}}$, then there exists a $q(n)$-qubit quantum state $\ket{\psi}$ such that
        \begin{equation}
            \Pr[V_n\left(\ket{x} \otimes \ket{0}^{\otimes r(n)} \otimes \ket{\psi}\right) \text{ accepts } ] \geq 2/3, \text{and}
        \end{equation}
        \item \textbf{Soundness:} If $x$ is a no-instance, i.e. $x \in A_{\mathrm{NO}}$, then for any $q(n)$-qubit quantum state $\ket{\psi}$, 
        \begin{equation}
            \Pr[V_n\left(\ket{x} \otimes \ket{0}^{\otimes r(n)} \otimes \ket{\psi}\right) \text{ accepts }] \leq 1/3.
        \end{equation}
    \end{enumerate}
\end{definition}

We will also use the standard notion of completeness for promise-problem complexity classes.  In our setting, completeness is used to formalize the sense in which the local Hamiltonian problem captures the full difficulty of $\QMA$.

\begin{definition}[Completeness]
    We say that a promise problem $A = \left(A_{\mathrm{YES}}, A_{\mathrm{NO}}\right)$ is complete for the class $\mathcal{C}$, notated as $\mathcal{C}$-complete if:
    \begin{enumerate}
        \item $A \in \mathcal{C}$; and
        \item for any promise problem $B = \left(B_{\mathrm{YES}}, B_{\mathrm{NO}}\right) \in \mathcal{C}$, there exists a polynomial-time reduction $f$ such that:
        \begin{enumerate}
            \item $x \in B_{\mathrm{YES}} \: \Longrightarrow \: f(x) \in A_{\mathrm{YES}}$; and
            \item $x \in B_{\mathrm{NO}} \: \Longrightarrow \: f(x) \in A_{\mathrm{NO}}$.
        \end{enumerate}
    \end{enumerate}
\end{definition}

Finally, we record the standard hardness statement for the $2$-local Hamiltonian problem in the relative-gap notation used throughout this paper.  This result provides the baseline $\QMA$-hardness result against which our $\NP$ containment for almost commuting Hamiltonians should be compared.

\begin{theorem}[\cite{kkr}]
    The \kLH\ problem is $\QMA$-complete when $\gamma = O(m^{-c})$ for any $c > 0$.
\end{theorem}
\begin{proof} The result shown in \cite{kkr} is that the problem is $\QMA$-complete for some particular $c$, but a simple gap-amplification argument\footnote{We have heard it credited to Lijie Chen, but it may be folklore.} gives the claimed statement. We illustrate this argument with some concrete numbers. Suppose LH is $\QMA$-hard for a relative gap of $\gamma = m^{-3} $, meaning an absolute gap in energy of $m \gamma$. Then consider a $k$-fold amplified Hamiltonian acting on a system of $kn$ qubits divided into $k$ groups of $n$ qubits, given by $H_k = \sum_{i=1}^{k} (H)_i \otimes I_{\bar{i}}$ . (That is, $H_k$ is a sum of copies of $H$ acting on the $i$th group of $n$ qubits, for $i$ running from $1$ to $k$.) This Hamiltonian has $m' = km$ terms and an absolute energy gap of $km\gamma$, so the new relative gap is $\gamma'(m') = m^{-3} = (m'/k)^{-3}$. If we set $k = m^{99}$---which gives $m' = m^{100}$ and $k = m'^{99/100}$---then $m'/k = m'^{0.01}$, so $\gamma'(m') =(m')^{-0.03}$. This establishes that the LH problem for $\gamma = m^{-0.03}$ is $\QMA$-hard.
\end{proof}

\begin{theorem}[\cite{bravyi2004commutativeversionklocalhamiltonian}]
\label{thm:clh-in-np}
    The $\gamma$-2-CLH promise problem is contained in $\NP$ for any admissible\footnote{Our definition of $\gamma$-2-CLH requires $\gamma \geq 1/\poly(n)$.} $\gamma$.
\end{theorem}
\begin{proof}
    Lemma~5 of~\cite{bravyi2004commutativeversionklocalhamiltonian} shows that the problem $\gamma$-2-CLH for any value of $\gamma$ is contained in $\NP$ if the "common eigenspace problem" (CES) for 2-local projectors is contained in $\NP$. This latter containment is shown in Theorem~3 of~\cite{bravyi2004commutativeversionklocalhamiltonian}.
\end{proof}

\section{Commutativity in the World of Qubit Operators}

In this section, we collect the qubit-specific facts that make the
pinching step possible.  We first define the pinching and snapping
operations for single-qubit operators, and then prove the elementary norm
bounds that control the errors they introduce.  We also isolate the
transitivity phenomenon for commuting qubit operators, which is one of the
main algebraic features that fails in higher local dimension and is used
crucially in our rounding argument.

\subsection{Pinching by a Hermitian Qubit Operator}

\begin{definition}[Pinching operator]\label{def:pinching-operator}
Let $A \in \mathrm{Herm}(\mathbb{C}^{2})$ be Hermitian with $\Delta(A)>0$ and spectral decomposition 
\begin{align}
A = \lambda_{\min}(A)\,\Pi + \lambda_{\max}(A)\,(I-\Pi),
\end{align}
where $\Pi$ and $(I-\Pi)$ are the eigenspace projectors. Then, the pinching operator $\mathcal{P}_A: \mathrm{L}(\mathbb{C}^{2})\rightarrow \mathrm{L}(\mathbb{C}^{2})$ associated with the eigenspace projectors of $A$ is defined as
\begin{align}
    \mathcal{P}_A(X) = \Pi \, X \, \Pi + (I-\Pi) X (I-\Pi)\, , \quad\quad \text{for } X \in \mathbb{C}^{2 \times 2}.
    \end{align}
For a scalar operator $A=cI$, we define $\mathcal{P}_A=\mathrm{id}$.
    
\end{definition}

Equipped with \Cref{def:pinching-operator}, we now show the following result.

\begin{lemma}[Pinching by a Hermitian qubit operator]\label{thm:pinch}
Let $A \in \mathrm{Herm}(\mathbb{C}^{2})$ be a Hermitian operator with $\Delta(A)>0$ and spectral decomposition given by
\begin{align}
A = \lambda_{\min}(A)\,\Pi + \lambda_{\max}(A)\,(I-\Pi)
\end{align}
with eigenvalues $\lambda_{\min},\lambda_{\max} \in \mathbb{R}$ and eigenspace projectors $\Pi$ and $(I-\Pi)$. Let $\mathcal{P}_A$ be the pinching operator in \Cref{def:pinching-operator} with respect to $A$.
Then, for any $B \in \mathbb{C}^{2 \times 2}$, the pinched matrix $\mathcal{P}_A(B)$ satisfies
\begin{align}
[\mathcal{P}_A(B), A] = 0
\qquad \text{and} \qquad
\|B - \mathcal{P}_A(B)\|
= \frac{\|[B,A]\|}{\Delta(A)} \,,
\end{align}
where $\Delta(A) = \lambda_{\max}(A) - \lambda_{\min}(A)$ is the spectral gap of $A$.
\end{lemma}

\begin{proof}
Both the operator norm and the commutator are invariant under unitary conjugation. Hence, without loss of generality, we may assume that $A$ is diagonal in its eigenbasis:
\begin{align}
A =
\begin{pmatrix}
\lambda_{\min}(A) & 0 \\
0 & \lambda_{\max}(A)
\end{pmatrix} \quad \text{ and } \quad
\quad
\Pi =
\begin{pmatrix}
1 & 0 \\
0 & 0
\end{pmatrix}.
\end{align}
Expanding $B$ in the eigenbasis of $A$, we get
\begin{align}
B =
\begin{pmatrix}
B_{1,1} & B_{1,2} \\
B_{2,1} & B_{2,2}
\end{pmatrix}.
\end{align}
A direct computation shows that
\begin{align}
[B,A]
&= BA - AB \\
&= \bigl(\lambda_{\max}(A) - \lambda_{\min}(A)\bigr)
\begin{pmatrix}
0 & B_{1,2} \\
-B_{2,1} & 0
\end{pmatrix}
= \Delta(A) \cdot
\begin{pmatrix}
0 & B_{1,2} \\
-B_{2,1} & 0
\end{pmatrix}.
\end{align}
By definition of the pinching operator, we have
\begin{align}
\mathcal{P}_A(B)
=
\begin{pmatrix}
B_{1,1} & 0 \\
0 & B_{2,2}
\end{pmatrix},
\end{align}
which is diagonal in the eigenbasis of $A$, and hence $[\mathcal{P}_A(B),A]=0$. Moreover,
\begin{align}
B - \mathcal{P}_A(B)
=
\begin{pmatrix}
0 & B_{1,2} \\
B_{2,1} & 0
\end{pmatrix}.
\end{align}
Next, consider the Pauli-$Z$ unitary
\begin{align}
Z=\begin{pmatrix}
1 & 0 \\
0 & -1
\end{pmatrix}.  
\end{align}
It is now easy to see that
\begin{align}
B - \mathcal{P}_A(B) = \Delta(A)^{-1}\cdot Z \cdot [B,A].
\end{align}
Using the unitary invariance of the operator norm with respect to $Z$,
we obtain the equality
\begin{align}
\|B - \mathcal{P}_A(B)\|
= \frac{\|Z\cdot [B,A]\|}{\Delta(A)} = \frac{\|[B,A]\|}{\Delta(A)}.
\end{align}
\end{proof}

\subsection{Snapping Hermitian qubit operators to multiples of the Identity}

In this section, we formally introduce the mathematical background for the ``snapping technique''.

\begin{definition}[Snapping operation]\label{def:snapping-operator}
The snapping operation is the linear map
\begin{align}
\mathcal{S} : \mathrm{Herm}(\mathbb{C}^{2}) \rightarrow \mathrm{Herm}(\mathbb{C}^{2}),
\quad X \mapsto \mathcal{S}(X) := \frac{\mathrm{Tr}(X)}{2}\,I.
\end{align}
For a Hermitian qubit operator, $\mathrm{Tr}(X)/2=(\lambda_{\min}(X)+\lambda_{\max}(X))/2$ is the midpoint of its spectrum.
\end{definition}
\begin{lemma}[Snapping a Hermitian qubit operator to the identity]\label{lemma:snapping-hermitian-qubit-hermitian-operator}
Let $B \in \mathrm{Herm}(\mathbb{C}^{2})$ be any matrix with spectral decomposition
\begin{align}
B = \lambda_{\min}(B)\,\Pi + \lambda_{\max}(B)\,(I-\Pi),
\end{align}
and let $\mathcal{S}$ denote the snapping operator defined in \Cref{def:snapping-operator}. Then,
\begin{align}
\|B - \mathcal{S}(B)\|
= \frac{\lambda_{\max}(B) - \lambda_{\min}(B)}{2}
= \frac{\Delta(B)}{2},
\qquad \|\mathcal{S}(B)\|\leq\|B\|,
\end{align}
where $\Delta(B) = \lambda_{\max}(B) - \lambda_{\min}(B)$ is the spectral gap of $B$.
\end{lemma}

\begin{proof}
The operator norm is invariant under unitary conjugation. Hence, without loss of generality, we may assume that $B$ is diagonal in its eigenbasis:
\begin{align}
B =
\begin{pmatrix}
\lambda_{\min}(B) & 0 \\
0 & \lambda_{\max}(B)
\end{pmatrix}.
\end{align}
By definition of the snapping operation, with $c=(\lambda_{\min}(B)+\lambda_{\max}(B))/2$, we have
\begin{align}
\mathcal{S}(B) = cI,
\qquad
\mathcal{S}(B)-B =
\begin{pmatrix}
\Delta(B)/2 & 0 \\
0 & -\Delta(B)/2
\end{pmatrix}.
\end{align}
The operator norm of a Hermitian matrix equals the maximum absolute value of its eigenvalues. Consequently,
\begin{align}
\|B-\mathcal{S}(B)\|=\frac{\Delta(B)}{2},
\qquad
\|\mathcal{S}(B)\|=|c|
\leq\max\{|\lambda_{\min}(B)|,|\lambda_{\max}(B)|\}=\|B\|.
\end{align}
\end{proof}

\subsection{Transitivity of (Almost) Commutativity of Qubit Operators}

\begin{lemma}[Gapped spectra imply transitivity of commutativity]\label{lem:transitivity}
Let $A,B,C \in \mathrm{Herm}(\mathbb{C}^{2})$ and suppose that $B$ has non-zero spectral gap $\Delta(B)=\lambda_{\max}(B) - \lambda_{\min}(B) > 0$. Then, 
$$
[A,B] = 0 \quad \text{ and } \quad [B,C] = 0 \quad \text{ implies } \quad [A,C]=0.
$$
\end{lemma}
\begin{proof}
Since $\Delta(B) > 0$, the matrix $B$ admits an eigenbasis, unique up to phases and ordering, in which it is diagonal. Let $\{\Pi, \Pi^\perp\}$ be the orthogonal projectors onto the respective eigenspaces.
Because $A$ commutes with $B$, and $C$ commutes with $B$, we can use simultaneous diagonalization (see \Cref{fact:sim-diag}) and the fact\footnote{Without that condition, we can only conclude that $A$ and $B$ can be simultaneously diagonalized and that $B$ and $C$ can be simultaneously diagonalized, but not necessarily all three can be simultaneously diagonalized. If $B$ is a multiple of the identity, then any two non-commuting matrices will be a counterexample.} $\Delta(B) > 0$ to diagonalize all three matrices with respect to the same eigenspace projectors $\{\Pi, \Pi^\perp\}$ as follows:
\begin{equation}
    A = a_1 \Pi + a_2 \Pi^\bot
\end{equation}
\begin{equation}
    B = b_1 \Pi + b_2 \Pi^\bot
\end{equation}
\begin{equation}
    C = c_1 \Pi + c_2 \Pi^\bot
\end{equation}
Using the orthogonality of $\{\Pi, \Pi^\perp\}$, we then find that
\begin{align*}
    [A, C] &= \left[a_1 \Pi + a_2 \Pi^\bot,  c_1 \Pi + c_2 \Pi^\bot\right]\\
    &= a_1c_1 [\Pi, \Pi] + a_1 c_2 [\Pi, \Pi^\bot] + a_2 c_1 [\Pi^\bot, \Pi] + a_2 c_2 [\Pi^\bot, \Pi^\bot] = 0.
\end{align*}
\end{proof}

The exact transitivity statement above is useful only after the relevant
operators have been made exactly commuting.  In the rounding argument,
however, we typically know only that two local components almost commute
with a common gapped pivot.  The following lemma provides the corresponding
robust version: if $A$ and $C$ each almost commute with a sufficiently
gapped qubit operator $B$, then $A$ and $C$ must almost commute with each
other. 

\begin{lemma}[Gapped spectra imply transitivity of almost-commutativity]\label{lem:transitivity-almost}
    Let $A,B,C \in \mathrm{Herm}(\mathbb{C}^{2})$ and suppose that $B$ has non-zero gap $\Delta(B)=\lambda_{\max}(B) - \lambda_{\min}(B) \geq \xi>0$. Then, 
$$
\|[A,B]\| \leq \eps \quad \text{ and} \quad \|[B,C]\| \leq \eps \quad \text{ implies } \quad \|[A,C]\| \leq \frac{2\eps}{\xi}\Big(\|A\| + \|C\| + \frac{\eps}{\xi} \Big) .
$$
When $\|A\|\leq 1$, $\|C\|\leq 1$, and $\xi \geq  \eps >0$, the bound becomes \[\|[A, C]\| \leq \frac{4\eps}{\xi} + \frac{2\eps^2}{\xi^2} \leq \frac{6\eps}{\xi}.\]
\end{lemma}
\begin{proof}
 According to \Cref{thm:pinch}, there exist (via pinching by $B$) $A'$ and $C'$ such that $[A', B] = [B, C'] = 0$ and $\|A - A'\| \leq \frac{\eps}{\xi}$ and $\|C - C'\| \leq \frac{\eps}{\xi}$. By \Cref{lem:transitivity}, it holds that $[C', A'] = 0$. Thus, we can bound $\|[A, C]\|$ as follows:
    \begin{align*}
        \|[A,C]\| & \leq 2 \cdot \|A - A'\| \cdot \|C\| + \|[A', C]\| \qquad \qquad \text{(via \Cref{lem:transitivity-of-almost-commutatitivty})}\\
        &\leq \frac{2\epsilon}{\xi} \cdot \|C\| + \|[C, A']\| \qquad\qquad(\|A - A'\| \leq \frac{\eps}{\xi} \text{ and } [C,A'] = -[A',C])\\
        &\leq \frac{2\epsilon}{\xi} \cdot \|C\| + 2 \cdot \|C - C'\| \cdot \|A'\| + \|[C',A']\|  \qquad\qquad \text{(via \Cref{lem:transitivity-of-almost-commutatitivty})}\\
        & \leq \frac{2\epsilon}{\xi} \cdot \|C\| + \frac{2\eps}{\xi} \cdot \|A'\| \qquad \qquad (\|C - C'\| \leq \frac{\eps}{\xi} \text{ and } [C', A'] = 0)\\
        & \leq \frac{2\eps}{\xi} \left(\|C\| + \|A\| + \eps/\xi \right) \qquad  \qquad (\|A'\| \leq \|A\| + \eps/\xi).
    \end{align*}
\end{proof}

\section{Local-to-Global Framework for Operators on $2$ Qubits}

In this section, we develop the local-to-global tools used in the rounding
argument.  The main point is that approximate commutation of two-qubit
Hamiltonian terms can be detected at the level of the single-qubit
operators appearing in their local Pauli decompositions.  This allows us to
translate global commutator bounds into local algebraic constraints, which
will later be used to choose pivots, perform snapping, and preserve locality
throughout the rounding procedure.

\subsection{Local Pauli Decomposition}
The following simple fact gives a decomposition of a two-qubit operator acting on qubits $s$ and $r$ into a sum of tensor product operators. Similar decompositions play a deep role in the theory of commuting Hamiltonians, for instance in Equation~(1) of~\cite{aharonov2011complexity}.

\begin{lemma}[Local Pauli decomposition]\label{algebra-fact-pauli-decomposition}
    Let $h_{s,r} \in \mathrm{L}(\mathcal{H}_s \otimes \mathcal{H}_r)$ be a linear two-qubit operator acting on the qubits labeled by $s$ and $r$. Then, $h_{s,r}$ can be decomposed as a linear combination of the form
    \begin{equation}\label{eqn:local-pauli-decomposition}
        h_{s,r} = \sum\limits_{\alpha = 0}^{3} A_s^{(\alpha)} \otimes \sigma_r^{(\alpha)}
    \end{equation}
where $\{\sigma_r^{\alpha}\}_{\alpha=0}^3$ are the Pauli matrices on the $r$-th qubit, and where $\{A_s^{(\alpha)}\}_{\alpha=0}^3$ is an ensemble of $2 \times 2$ complex matrices. Furthermore, the components $\{A_s^{(\alpha)}\}_{\alpha=0}^3$ are Hermitian if and only if $h_{s, r}$ is Hermitian.
\end{lemma}
\begin{proof}
Here, we make use of \Cref{fact:pauli-properties}; namely, that the Pauli matrices form an orthogonal basis of $\mathrm{L}(\mathbb{C}^{2})$ with respect to the Hilbert-Schmidt inner product. Observing that tensor products of Pauli matrices form a basis 
of the space of linear operators on the tensor system $\mathbb{C}^{2} \otimes \mathbb{C}^{2}$, we obtain the decomposition
    \[ h_{s,r} = \sum_{\alpha=0}^{3} \sum_{\beta=0}^{3} c_{\alpha, \beta} \, \sigma_s^{(\beta)} \otimes \sigma_r^{(\alpha)} \, .\]
Next, we let $A^{(\alpha)}_s = \sum_{\beta=0}^{3} c_{\alpha, \beta} \, \sigma_s^{(\beta)}$ and show the \emph{if and only if} statement. First, if $A^{(\alpha)}_s$ is Hermitian, then its tensor product with a Hermitian Pauli matrix must also be Hermitian, and the sum $h_{s,r}$ over these tensor products is going to be Hermitian. Second, to show that if $h_{s,r}$ is Hermitian, then each $A^{(\alpha)}_s$ is Hermitian, we proceed as follows; first, observe that
\begin{align}
\mathrm{Tr}_{r}\left[\left(I \otimes \sigma_r^{(\alpha)} \right) h_{s,r} \right] &= \mathrm{Tr}_{r} \left[\sum\limits_{\beta = 0}^{3} A_s^{(\beta)} \otimes \sigma_r^{(\alpha)} \sigma_r^{(\beta)}\right] && \text{(by \Cref{eqn:local-pauli-decomposition})}\\
    &= \sum\limits_{\beta = 0}^{3} A_s^{(\beta)} \cdot \mathrm{Tr}\left(\sigma_r^{(\alpha)} \sigma_r^{(\beta)}\right) && \text{(by linearity of the trace)}\\
    &= \sum\limits_{\beta = 0}^{3} 2\delta_{\alpha, \beta}\ A_s^{(\beta)}&& (\text{by \Cref{fact:pauli-properties}})\\
    &= 2 A_s^{(\alpha)}.
\end{align}

\noindent We can similarly compute $\mathrm{Tr}_{r}\left[ h_{s,r} \left(I \otimes \sigma_r^{(\alpha)} \right)\right] = 2 A_s^{(\alpha)}$. Combining the two, we get
\begin{equation}\label{eqn:A-S-Alpha}
    A_s^{(\alpha)} = \frac{1}{2} \mathrm{Tr}_{r}\left[\left(I \otimes \sigma_r^{(\alpha)} \right) h_{s,r} \right] = \frac{1}{2} \mathrm{Tr}_{r}\left[ h_{s,r} \left(I \otimes \sigma_r^{(\alpha)} \right)\right].
\end{equation}
Putting everything together, we can verify that $A_s^{(\alpha)}$ is Hermitian by noting the following.
\begin{align}
\left(A_s^{(\alpha)}\right)^\dagger &= \frac{1}{2} \mathrm{Tr}_{r}\left[\left(I \otimes \sigma_r^{(\alpha)} \right) h_{s,r} \right]^\dagger && \text{(by \Cref{eqn:A-S-Alpha})}\\
    &= \frac{1}{2} \mathrm{Tr}_{r}\left[ h_{s,r}^\dagger \left(I \otimes \sigma_r^{(\alpha)} \right)^\dagger \right] && \text{("Socks and Shoes" property)} \\
    &= \frac{1}{2} \mathrm{Tr}_{r}\left[ h_{s,r} \left(I \otimes \sigma_r^{(\alpha)} \right) \right] && \text{($h_{s,r}$ and $\sigma_r^{(\alpha)}$ are hermitian)}\\
    &= A_s^{(\alpha)} && \text{(by \Cref{eqn:A-S-Alpha})}
\end{align}
\end{proof}

When working with local Pauli decompositions of two-qubit operators, it will be convenient for us to use the following notion of a \emph{gapped} decomposition.
\begin{definition}[Gapped Pauli decomposition]\label{def:gapped-pauli-decomp}
  A two-qubit operator $h_{s,r} \in \mathrm{Herm}(\mathcal{H}_s \otimes \mathcal{H}_r)$ acting on the qubits labeled by $s$ and $r$ is said to be $\eta$-gapped on the $s$-subsystem (notated as $(s, \eta)$-gapped) if any of the local components $\{A_s^{(\alpha)}\}_{\alpha=0}^3$ which appear in the local Pauli decomposition of the form
      \begin{equation}
        h_{s,r} = \sum\limits_{\alpha = 0}^{3} A_s^{(\alpha)} \otimes \sigma_r^{(\alpha)}
    \end{equation}
  is $\eta$-gapped; in other words, there exists a matrix $A_s^{(\alpha^*)} \in \mathbb{C}^{2 \times 2}$ which is $\eta$-gapped such that $\Delta(A_s^{(\alpha^*)}) \geq \eta$.
    The property $\left(r, \eta\right)$-gapped is defined similarly. We say the operator is $\eta$-gapped on both subsystems if it is both $\left(r, \eta\right)$-gapped and $\left(s, \eta\right)$-gapped.
\end{definition}

\subsection{Global-to-Local Inheritance of Norms of Commutators}
Let $h_{s,r}$ and $h_{s,q}$ be $2$ Hamiltonian terms on qubit pairs $\{s,r\}$ and $\{s,q\}$, where $s,r,q$ are distinct, such that they $\eps$-almost commute i.e. $\|[h_{s,r}, h_{s,q}]\| \leq \eps$. Let their local components on the shared qubit be $\{A^{(\alpha)}\}$ and $\{B^{(\beta)}\}$. An important question for our techniques is: what can we say about the commutator of any pair $\|[A^{(\alpha)}, B^{(\beta)}]\|$? We will prove in~\Cref{thm:comm-propagates-down} that any pair of local components must also $\eps$-almost commute. For that proof, we will use~\Cref{lem:opnorm-orthog-components}.

\begin{restatable}[Operator norm lower bound from orthogonal components]{lemma}{OpNormOrthogCompLemma}
\label{lem:opnorm-orthog-components}
Let $\{M_i\}_i \subseteq \mathbb{C}^{d_1 \times d_1}$, and let
$\{N_i\}_i \subseteq \mathbb{C}^{d_2 \times d_2}$ be Hermitian binary
observables satisfying
\begin{align}
    N_i^2 = I_{d_2}
    \qquad \text{and} \qquad
    \mathrm{Tr}[N_i N_j] = d_2\,\delta_{i,j}.
\end{align}
Then, any operator of the form
$T = \sum_i M_i \otimes N_i$ satisfies
\begin{align}
    \|T\| \geq \|M_i\|, \qquad \text{ for every $i$.}
\end{align}
\end{restatable}

\begin{proof}
We compute
\begin{align}
    T^\dagger T
    = \sum_{i,j} M_i^\dagger M_j \otimes N_i N_j,
\end{align}
where we used that each $N_i$ is Hermitian.  Taking the partial trace over
the second subsystem and using the Hilbert-Schmidt orthogonality of the
$N_i$'s gives
\begin{align}
    \mathrm{Tr}_2[T^\dagger T]
    &= \sum_{i,j} M_i^\dagger M_j \, \mathrm{Tr}[N_i N_j] \\
    &= d_2 \sum_i M_i^\dagger M_i .
\end{align}
By \Cref{fact:op-norm-ptrace}, extended by homogeneity to positive semidefinite operators, we have
\begin{align}
    \|T\|^2
    = \|T^\dagger T\|
    \geq \frac{1}{d_2}\,\bigl\|\mathrm{Tr}_2[T^\dagger T]\bigr\|
    = \left\|\sum_i M_i^\dagger M_i\right\|.
\end{align}
Since each $M_j^\dagger M_j$ is positive semidefinite, we have
\begin{align}
    \sum_j M_j^\dagger M_j \succeq M_i^\dagger M_i
\end{align}
for every $i$, and hence
\begin{align}
    \left\|\sum_j M_j^\dagger M_j\right\|
    \geq \|M_i^\dagger M_i\|
    = \|M_i\|^2 .
\end{align}
Combining the two inequalities yields $\|T\|^2 \geq \|M_i\|^2$, and
therefore $\|T\| \geq \|M_i\|$.
\end{proof}
\begin{theorem}\label{thm:comm-propagates-down}
Let $h_{s,r}$ and $h_{s,q}$ be two $2$-qubit linear operators on qubit pairs $\{s,r\}$ and $\{s,q\}$, where $s,r,q$ are distinct. According to~\Cref{algebra-fact-pauli-decomposition}, these can be written as
\begin{equation}
        h_{s,r} = \sum\limits_{\alpha = 0}^{3} A_s^{(\alpha)} \otimes \sigma_r^{(\alpha)} \quad\quad\quad h_{s,q} = \sum\limits_{\beta = 0}^{3} B_s^{(\beta)} \otimes \sigma_q^{(\beta)}
    \end{equation}
for ensembles of $2 \times 2$ complex matrices $\{A_s^{(\alpha)}\}_{\alpha=0}^3$ and $\{B_s^{(\beta)}\}_{\beta=0}^3$.
Then, if the terms satisfy $\|[h_{s, r}, h_{s, q}]\| \leq \eps$, for some $\eps \geq 0$, this implies that $\|[A_s^{(\alpha)}, B_s^{(\beta)}]\| \leq \eps$, for all $\alpha, \beta$.
\end{theorem}
\begin{proof}
We can re-write the linear operators with an identity factor on the unaffected qubit,
    \begin{equation}
        h_{s,r} = \sum\limits_{\alpha = 0}^{3} A_s^{(\alpha)} \otimes \sigma_r^{(\alpha)} \otimes I_{q}\quad\quad\text{ and } \quad\quad h_{s,q} = \sum\limits_{\beta = 0}^{3} B_s^{(\beta)} \otimes I_r \otimes \sigma_q^{(\beta)}.
    \end{equation}
Then, we can compute the products $h_{s,r} h_{s,q}$ and $h_{s,q} h_{s,r}$ as follows.
\begin{equation}
    h_{s,r} h_{s,q} = \sum\limits_{\alpha, \beta} A_s^{(\alpha)} B_s^{(\beta)} \otimes \sigma_r^{(\alpha)} \otimes \sigma_q^{(\beta)}
\end{equation}
\begin{equation}
    h_{s,q} h_{s,r} = \sum\limits_{\alpha, \beta}  B_s^{(\beta)} A_s^{(\alpha)} \otimes \sigma_r^{(\alpha)} \otimes \sigma_q^{(\beta)}
\end{equation}
Then, we can plug these in the definition of the commutator.
\begin{align*}
    \left[h_{s,r}, h_{s,q}\right] &= h_{s,r} h_{s,q} - h_{s,q} h_{s,r}\\
    &= \sum\limits_{\alpha, \beta} \left(A_s^{(\alpha)} B_s^{(\beta)} - B_s^{(\beta)} A_s^{(\alpha)}\right) \otimes \sigma_r^{(\alpha)} \otimes \sigma_q^{(\beta)}\\
    &= \sum\limits_{\alpha, \beta} \left[A_s^{(\alpha)}, B_s^{(\beta)} \right] \otimes \left[\sigma_r^{(\alpha)} \otimes \sigma_q^{(\beta)}\right]
\end{align*}
Using \Cref{fact:pauli-properties}, we can conclude that $\{\sigma_r^{(\alpha)} \otimes \sigma_q^{(\beta)}\}_{\alpha,\beta}$ forms a set of pairwise orthogonal observables on the tensor product system $\mathcal{H}_r \otimes \mathcal{H}_q$. Thus, for each $\alpha, \beta$, \Cref{lem:opnorm-orthog-components} implies that 
\begin{align}
\|[A_s^{(\alpha)}, B_s^{(\beta)}]\| \leq \|[h_{s, r}, h_{s, q}]\| \leq \eps.
\end{align}
This proves the claim.
\end{proof}

\subsection{Snapping a $2$-local Hamiltonian term into $1$-local Hamiltonian term}

\begin{theorem}[Snapping $2$-local Hamiltonian terms into $1$-local terms]
\label{thm:snapping-2-local-hamiltonian-term}
   Let $h_{s,r} \in \mathrm{Herm}(\mathcal{H}_s \otimes \mathcal{H}_r)$ be a two-qubit operator acting on qubits $s$ and $r$ which is not $(s,\eta)$-gapped according to \Cref{def:gapped-pauli-decomp}. Then, there exists a $1$-local Hamiltonian term $\hat{h}_{r}$ such that:
    \begin{equation}
    \label{eqn:snapped-distance-between-hamiltonian-terms}
        \|h_{s,r} - I_s\otimes\hat{h}_{r}\| \leq 4 \eta, \qquad \|\hat{h}_{r}\| \leq \|h_{s,r}\|.
    \end{equation}
\end{theorem}
\begin{proof}
Using~\Cref{algebra-fact-pauli-decomposition}, we can expand the two-qubit operator $h_{s,r}$ as
    \begin{align*}
        h_{s,r} = \sum\limits_{\alpha = 0}^{3} A_{s}^{\alpha} \otimes \sigma_{r}^{(\alpha)}
        = \sum\limits_{\alpha = 0}^{3} \Big(\lambda_{\max}^{\alpha} \Pi_{s}^{\alpha} + \lambda_{\min}^{\alpha} \left(I - \Pi_{s}^{\alpha}\right)\Big) \otimes \sigma_{r}^{(\alpha)}.
    \end{align*}
The hypothesis that $h_{s,r}$ is not $(s,\eta)$-gapped means that for all $\alpha$, $\Delta(A_{s}^{\alpha}) < \eta$. Applying the snapping operation of \Cref{def:snapping-operator}, we replace each coefficient by the midpoint of its spectrum:
\[
\hat{A}_s^{\alpha} := \mathcal{S}(A_s^{\alpha}) = \frac{\mathrm{Tr}(A_s^{\alpha})}{2}I_s
= \frac{\lambda_{\min}(A_s^{\alpha})+\lambda_{\max}(A_s^{\alpha})}{2}I_s.
\]
The induced map on the full two-qubit term is norm-nonincreasing. Define the operator on qubit $r$ by
\begin{align}
    \hat{h}_r := \sum\limits_{\alpha = 0}^{3} \frac{\mathrm{Tr}(A_s^{\alpha})}{2}\,\sigma_r^{(\alpha)}
    = \frac{1}{2}\mathrm{Tr}_s(h_{s,r}).
\end{align}
Indeed, conjugating a single-qubit operator by each of $I,X,Y,Z$ and averaging removes its traceless Pauli components. Thus,
\[
I_s\otimes\hat{h}_r
= \frac{1}{4}\sum_{\mu=0}^{3}(\sigma_s^{(\mu)}\otimes I_r)h_{s,r}(\sigma_s^{(\mu)}\otimes I_r).
\]
Each $\sigma_s^{(\mu)}\otimes I_r$ is unitary. By the triangle inequality and unitary invariance,
\[
\|\hat{h}_r\|=\|I_s\otimes\hat{h}_r\|
\leq\frac14\sum_{\mu=0}^{3}\|h_{s,r}\|=\|h_{s,r}\|.
\] Since $\|A_s^{\alpha}-\hat{A}_s^{\alpha}\|=\Delta(A_s^{\alpha})/2$, we can calculate the rounding error explicitly as follows:
\begin{align*}
    \|h_{s,r} - I_s\otimes\hat{h}_{r}\| &= \left\|\sum\limits_{\alpha = 0}^{3} \left( A_s^{\alpha} - \hat{A}_s^{\alpha} \right)\otimes \sigma_{r}^{(\alpha)}\right\|\\
    &\leq \sum\limits_{\alpha = 0}^{3} \left\|\left( A_s^{\alpha} - \hat{A}_s^{\alpha} \right)\otimes \sigma_{r}^{(\alpha)} \right\| & \text{(triangle inequality)}\\
    &\leq \sum\limits_{\alpha = 0}^{3} \left\| A_s^{\alpha} - \hat{A}_s^{\alpha}  \right\| \cdot \underbrace{\| \sigma_r^{(\alpha)} \|}_{=1} & (\text{by the properties of } \|\cdot\|)\\
    &\leq \frac{1}{2}\sum\limits_{\alpha = 0}^{3} \Delta(A_{s}^{\alpha}) \,\leq\, 2\eta \,\leq\, 4 \eta.
\end{align*}
\end{proof}

\section{Rounding Almost Commuting $2$-Local Qubit Hamiltonians}
In this section, we prove the main result of our paper: any almost commuting $2$-local Hamiltonian (on qubits) can be mapped to a nearby exactly commuting (qubit) Hamiltonian of the same locality. The proof of our main result uses a technical helper lemma (\Cref{lemma:double-extra-cheese}), which we formally state and prove later in this section.

\begin{theorem}[Algorithmic rounding for almost commuting $2$-local Hamiltonians on qubits]\label{thm:transforming-hamiltonian}\ \\
Let $H = \sum\limits_{I \in \mathcal{I}} h_{I}$ be a $2$-local 
Hamiltonian on $n$ qubits, where $\mathcal{I}$ is a collection of distinct sets of indices affected by each Hamiltonian term, where $m = |\mathcal{I}|$ is the total number of terms, and where
\begin{itemize}
    \item each term $h_{I}$, for $I \in \mathcal{I}$, is a Hermitian operator that acts as non-identity on at most two qubits (i.e. $|I| \leq 2$), of unit norm $\| h_I \| \leq 1$, and
    \item all terms pair-wise $\eps$-almost commute with $0 \leq \eps \leq 1$ such that
\[ \| [h_{I}, h_{J} ] \| \leq \eps, \quad\quad \forall I, J \in \mathcal{I}. \]
\end{itemize}
Given a description of $H$ and the bound $\eps$, it is possible to compute in classical deterministic time $O(m)$\footnote{And polynomial in the bit complexity of each Hamiltonian term and of $\eps$.} the description of a nearby exactly commuting Hamiltonian $\hat{H} = \sum\limits_{I \in \mathcal{I}} \hat{h}_{I}$ such that 
\begin{itemize}
\item $\hat{H}$ is $2$-local, thereby preserving the locality of the Hamiltonian $H$, and $\|\hat{h}_I\|\leq\|h_I\|\leq 1$ for every $I\in\mathcal{I}$,
\item $\hat{H}$ is a qubit Hamiltonian, thereby preserving the local dimension of $H$,
\item all of the terms in $\hat{H}$ pairwise commute such that
    \begin{align} 
[\hat{h}_{I}, \hat{h}_{J}] &= 0, \quad\quad \forall I, J \in \mathcal{I}, \label{eq:exact-commutation}
\end{align}
\item all of the terms in $\hat{H}$ are (point-wise) nearby in the sense that 
\begin{align}
\| \hat{h}_{I} - h_{I} \| \leq 200\ \eps^{1/3}, \quad\quad \forall I \in \mathcal{I}, \text{ and}\label{eq:overall-rounding-term-error}
\end{align}
\item $\hat{H}$ is overall close to $H$ in operator norm such that
\begin{align}
\|\hat{H} - H\| \leq 200\ m \ \eps^{1/3}.\label{eq:rounding-error}
\end{align}
\end{itemize}
\end{theorem}

\begin{proof}
If $\eps=0$, set $\hat{H}=H$, which satisfies all the claims. Henceforth assume $\eps>0$.
We remark that the statement of the theorem allows the Hamiltonian to have terms that have locality strictly less than $2$. For convenience, we assume that the Hamiltonian is of the form $H = \sum_{(i,j) \in E} h_{i,j}$ so that all terms in $H$ are formally two-local. This avoids additional bookkeeping and is purely for clarity of the exposition (it does not materially affect any of the arguments). The ordering of the two indices on the Hamiltonian term is irrelevant, so we will freely write $h_{i,j}$ or $h_{j,i}$ for the same term as convenient.

Our rounding procedure generates the Hamiltonian $\hat{H}$ as follows. We will have two stages of snapping with gap parameters $\eta_2$ in the first stage and $\eta_1$ in the second stage (note that the index is counting down), with $\eta_1=\eta_2=\eps^{1/3}$. Similarly, the commutator norm will be bounded by $\eps_2$ in the first stage and $\eps_1$ in the second stage (with the index counting down as well).

\begin{enumerate}

    \item Define sets of terms $T_2, T_1, T_0$ that will keep track of 2-local, 1-local, and 0-local terms, respectively, through the snapping process. To start out, let $T_2$ contain all the terms in $H$, and let $T_1$ and $T_0$ be empty.
    
    Let $\eps_2 = \eps$.
    \item \textbf{Snapping $2$-local terms.} Let $\eta_2 := \eps_2^{1/3}$. For every term $h_{i, j}$ in $T_2$:
    \begin{enumerate}
        \item If $h_{i, j}$ is $\eta_2$-gapped on both subsystems according to~\Cref{def:gapped-pauli-decomp} (i.e. ($i, \eta_2$)-gapped and ($j, \eta_2$)-gapped), remove it from $T_2$ and place it in $H^\snp_2$.
        \item Else, pick the qubit $j$ (WLOG) such that $h_{i, j}$ is not $(j, \eta_2)$-gapped. Snap using~\Cref{thm:snapping-2-local-hamiltonian-term} into a $1$-local term on the subsystem $i$ to obtain $h'_{i,j}$ as below and place $h'_{i,j}$ in $T_1$.
        \begin{equation}
        h'_{i,j} = g_i \otimes I_j  \quad\quad \text{ where }\quad\quad     \| h'_{i,j} - h_{i,j} \| \leq 4\eta_2.
        \end{equation}
    \end{enumerate}
    \item \textbf{Snapping $1$-local terms.} Let $\eta_1 := \eta_2$. For each term $h'_{i,j}$ in $T_1$ ($1$-local terms), perform the following.
    \begin{enumerate}
        \item If $h'_{i,j}$ is $(i, \eta_1)$-gapped, place the term into $H^\snp_1$.
        \item Else, snap that term into a multiple of the identity $h''_{i,j} := \mathcal{S}(g_i)\otimes I_j = c I$, where $c=\mathrm{Tr}(g_i)/2$, and place that term into $T_0$.
        \begin{equation}
            \|h''_{i,j} - h'_{i,j}\| \leq 4 \eta_1.
        \end{equation}
    \end{enumerate}
    \item Move every term in $T_0$ to $H^\snp_0$.
    \item We now have:
    \begin{equation}
        H^\snp = H^\snp_2 + H^\snp_1 + H^\snp_0
    \end{equation}
    Both snapping stages are norm-nonincreasing: this follows from \Cref{thm:snapping-2-local-hamiltonian-term} for the first stage and from $\|\mathcal{S}(g_i)\|=|\mathrm{Tr}(g_i)|/2\leq\|g_i\|$ for the second. Thus $\|h^\snp_{i,j}\|\leq\|h_{i,j}\|\leq 1$ for every term. Moreover:
    \begin{enumerate}
    \item \textbf{Spectral gap:}
    \begin{enumerate}
        \item Every term $h_{i,j}^\snp$ in $H^\snp_2$ is $\eta_2$-gapped on both subsystems $i$ and $j$ according to \Cref{def:gapped-pauli-decomp}.
        \item Every term $h_{i,j}^\snp$ in $H^\snp_1$ is $\eta_1$-gapped (according to \Cref{def:gapped-pauli-decomp}) on the subsystem acted upon nontrivially.
        \item Every term $h_{i,j}^\snp$ in $H^\snp_0$ is proportional to the identity (zero spectral gap).
    \end{enumerate}
    \item \textbf{Rounding error:}
    \begin{enumerate}
        \item Every term $h^\snp_{i,j}$ in $H^\snp_2$ is identical to the corresponding term $h_{i,j}$ in $H$.
        \item For every term $h^\snp_{i,j}$ in $H^\snp_1$, we have
        \begin{equation}\label{eqn:distance-snapped-once}
            \| h^\snp_{i,j} - h_{i,j} \| \leq 4\eta_2.
        \end{equation}
        This is because every term in $H^\snp_1$ was obtained by one application of snapping with gap parameter $\eta_2$.
        \item For every term $h^\snp_{i,j}$ in $H^\snp_0$, we have
        \[ \|h^\snp_{i,j} - h_{i,j}\| \leq 4\eta_2 + 4\eta_1. \]
        This is because every such term was obtained from snapping an original two-local term twice, first with gap parameter $\eta_2$ and then with gap parameter $\eta_1$.
    \end{enumerate}
        \item \textbf{Commutativity error:} Let $\eps_1 := \eps$.
    Write $\mathcal{S}_s(X)=\frac12 I_s\otimes\mathrm{Tr}_s(X)$ for snapping qubit $s$, with tensor factors ordered appropriately. As shown in \Cref{thm:snapping-2-local-hamiltonian-term}, this map is an average of Pauli conjugations and is therefore contractive in operator norm. If $Y$ acts trivially on $s$, then
    \[
        [\mathcal{S}_s(X),Y]=\mathcal{S}_s([X,Y]).
    \]
    Recall that the sets of indices are distinct, so two sets of indices share at most one qubit. If snapping removes that shared qubit from either term, the resulting terms act on non-overlapping qubits and commute exactly. Otherwise, snapping acts only on qubits private to each term, so the identity above and contractivity apply. In particular:
    \begin{enumerate}
        \item $H^\snp_2$ contains $2$-local terms such that each pair of terms $\eps_2$-commute.
        \item $H^\snp_1$ contains $1$-local terms such that each pair of terms $\eps_1$-commute. For two such terms overlapping on qubit $i$, write $h^\snp_{i,j}=\mathcal{S}_j(h_{i,j})$ and $h^\snp_{i,k}=\mathcal{S}_k(h_{i,k})$, where $j\neq k$. Then
        \[
            \|[h^\snp_{i,j},h^\snp_{i,k}]\|
            =\|\mathcal{S}_j\mathcal{S}_k([h_{i,j},h_{i,k}])\|
            \leq\eps=\eps_1.
        \]
        \item Any term in $H^\snp_2$ $\eps_1$-commutes with any term in $H^\snp_1$. Their commutator is zero unless they overlap on the qubit affected by the $1$-local term. For $h^\snp_{i,j}=h_{i,j}\in H^\snp_2$ and $h^\snp_{i,k}=\mathcal{S}_k(h_{i,k})\in H^\snp_1$, we have
        \[
            \|[h^\snp_{i,k},h^\snp_{i,j}]\|
            =\|\mathcal{S}_k([h_{i,k},h_{i,j}])\|
            \leq\eps=\eps_1.
        \]
    \end{enumerate}
    Terms in $H^\snp_0$ are scalar and commute with every term.
    \end{enumerate}
    \item \textbf{Pivot selection.} For each qubit $i \in [n]$, assign it a pivot operator $R_i$ as follows.
    \begin{enumerate}
        \item If qubit $i$ is affected by at most one non-identity Hamiltonian term in $H^\snp$, let $R_i$ be the one-qubit identity operator.
        \item Otherwise, there are $2$ subcases:
        \begin{enumerate}
            \item At least one of the terms $h^\snp_{i,j}$ in $H_1^\snp$ acts nontrivially on $i$. This term is $1$-local: $h^\snp_{i,j} = g_i \otimes I_j$ for some $g_i$. In this case, set $R_i = g_i$.
            \item All of the terms in $H^\snp$ acting nontrivially on $i$ are $2$-local: in this case, pick an arbitrary one of these terms $h^\snp_{ij}$ (e.g. the one with the lexicographically first index pair $i,j$). By the assumption that $h^\snp_{ij} \in H^\snp_2$, it must be $\eta_2$-gapped on both systems $i$ and $j$. So writing its Pauli decomposition as
            \[ h^\snp_{ij} = \sum_{\alpha = 0}^{3} A^\alpha_i \otimes \sigma^\alpha_j, \]
            we let $R_i$ be any of the operators $A^\alpha_i$ that has spectral gap at least $\eta_2$.
        \end{enumerate}

    \end{enumerate}
    We claim that for each qubit $i \in [n]$ there exists a choice of parameters $\kappa_i, \zeta_i$ such that the following hold.
    \begin{enumerate}
        \item[1.] For any Hamiltonian term $h^\snp_{i,j} \in H^\snp$,
        \[ \|[ h^\snp_{i,j}, R_i \otimes I_j] \| \leq \kappa_i, \]
        and likewise with $i$ and $j$ interchanged.
        \item[2.] For each $i$, either $R_i$ is identity or the spectral gap of $R_i$ is at least $\zeta_i$.
        \item[3.] For each $i$, $\kappa_i/\zeta_i \leq \max\{ \eps_1 /\eta_1, 24\eps_2/\eta_2^2\}$.
    \end{enumerate}
    
    \paragraph{Case 1: at most 1 term in $H^\snp$ acts nontrivially on $i$.} In this case, the algorithm above sets $R_i = I$. Let $\kappa_i = 0$ and $\zeta_i = 1$. The three properties hold.

\paragraph{Case 2: at least 2 distinct terms in $H^\snp$ act nontrivially on $i$.}

Within this case, there are two subcases.

\subparagraph{Case 2a: at least one of the terms in $H^\snp$ acting nontrivially on $i$ is 1-local.}

In this case, the procedure above sets $R_i$ to be equal to (the nontrivial tensor factor of) some such 1-local term $h^\snp_{i,j}$. Let $\kappa_i$ be equal to $\eps_1$ and $\zeta_i = \eta_1$. We know that the term $h^\snp_{i,j}$ has commutativity error $\eps_1$ with all other terms in $H^\snp$ and has spectral gap at least $\eta_1$, so the three properties hold.

\subparagraph{Case 2b: all of the terms in $H^\snp$ acting nontrivially on $i$ are 2-local.}

In this case, the procedure above will take $R_i$ to be some operator from the local decomposition of a term $h^\snp_{i,k}$. We will set $\kappa_i = 24\eps_2/\eta_2$ and $\zeta_i = \eta_2$. This setting automatically satisfies the third property.

To see that the properties hold with these parameters, let the decomposition of $h^\snp_{i,k}$ be
\begin{equation}
    h^\snp_{i,k} = \sum_{\beta=0}^3 B_i^\beta \otimes \sigma_k^\beta.
\end{equation}
Let $\beta^*$ be such that $R_i = B_i^{\beta^*}$. By construction, we know that $\beta^*$ was chosen so that the spectral gap of  $R_i = B_i^{\beta^*}$ is at least $\eta_2$, which establishes the second property. 

We now show the first property. Let $h^\snp_{i,j}$ be a Hamiltonian term acting on qubit $i$. Now, suppose $j \neq k$. Then decompose $h^\snp_{i,j}$ as
\begin{equation}
    h^\snp_{i,j}  = \sum_{\alpha=0}^3 A_i^\alpha \otimes \sigma_j^\alpha. 
\end{equation}
Since $\|[h^\snp_{i,j},h^\snp_{i,k}]\| \leq \eps_2$, we can conclude by applying \Cref{thm:comm-propagates-down} that for each $\alpha$:
\begin{equation}
 \| [A_i^\alpha, R_i] \| \leq \eps_2, \label{eq:pivot-commutes-different-term-2}
\end{equation}
and therefore
\begin{equation}
    \| [h^\snp_{i,j}, R_i \otimes I_j ] \| \leq 4 \eps_2 \leq 24\eps_2/\eta_2,
\end{equation}
under the assumption that $\eta_2 \leq 6$.

Now, suppose $j = k$. In this case, we are going to use transitivity (\Cref{lem:transitivity-almost}) to obtain our conclusion. Pick $j' \neq k$ such that $h^\snp_{i,j'}$ acts nontrivially on qubit $i$. (We are guaranteed that such a $j'$ exists by the assumption that we are in Case 2.) Such Hamiltonian term has the decomposition
\begin{equation}
    h^\snp_{i,j'} = \sum_{\alpha=0}^3 A_i^\alpha \otimes \sigma_{j'}^\alpha. 
\end{equation}
By \Cref{eq:pivot-commutes-different-term-2}, we know that for each $\alpha$, $\| [A_i^\alpha, R_i]\| \leq \eps_2.$ Now, let $B_i^\beta$ be an operator from the decomposition of $h^\snp_{i,k}$. Since $\| [h^\snp_{i,k}, h^\snp_{i,j'}] \| \leq \eps_2$, we can conclude by \Cref{thm:comm-propagates-down} that for any $\alpha$:
\begin{equation}
    \| [ B_i^\beta, A_i^\alpha] \| \leq \eps_2.
\end{equation}

Furthermore, by the fact that we chose $h_{ij'}$ to act nontrivially on qubit $i$, we know that it is $\eta_2$-gapped on qubit $i$, and therefore that there exists an $\alpha^*$ such that $\Delta(A_i^{\alpha^*}) \geq \eta_2$. 

Now, we will use transitivity: we know that $\|[B_i^\beta, A_i^{\alpha^*}]\| \leq \eps_2$ and $\|[A_i^{\alpha^*}, R_i]\| \leq \eps_2$, so by \Cref{lem:transitivity-almost}, we can conclude that 
\begin{equation}
    \| [B_i^\beta, R_i] \| \leq \frac{6\eps_2}{\Delta(A_i^{\alpha^*})} \leq \frac{6 \eps_2}{\eta_2},
\end{equation}
where we are assuming that $\eps_2 \leq \eta_2$ and moreover that the terms of the Hamiltonian have operator norm at most $1$. So, we have
\[ \| [h^\snp_{i,k}, R_i \otimes I_k ] \| \leq \frac{24\eps_2}{\eta_2}. \]
    \item \textbf{Pinch everything.} Define the global pinching map $\cP$ that pinches each qubit $i$ by $R_i$. Our final Hamiltonian will be
    \[ \widehat{H} = \cP(H^\snp).\]
    Tensor-product pinching is an average of unitary conjugations and is therefore contractive in operator norm. Consequently, $\|\hat{h}_{i,j}\|\leq\|h^\snp_{i,j}\|\leq\|h_{i,j}\|\leq 1$.
    By \Cref{lemma:double-extra-cheese}, we conclude that the resulting Hamiltonian $\widehat{H}$ is commuting, and that for each term $\widehat{h}_{i,j}$ in $\widehat{H}$, its distance from the corresponding term $h^\snp_{i,j}$ in $H^\snp$ is bounded by 
    \[ \|\hat{h}_{i,j} - h^\snp_{i,j} \| \leq  4\max_{i,j} \left\{ \left(\frac{\kappa_i}{\zeta_i} + \frac{\kappa_j}{\zeta_j} \right) \right\} \leq 8 \max\left\{ \frac{\eps_1}{\eta_1}, \frac{24\eps_2}{\eta_2^2} \right\} \leq 192\ \eps^{1/3}, \]
where the final inequality follows from the parameter settings we have made:
    \begin{align}
        \eps_2 &= \eps,\\
        \eps_1 &= \eps,\\
        \eta_2 &= \sqrt[3]{\eps}, \text{and}\label{ineq:eta_2}\\
        \eta_1 &= \eta_2 = \sqrt[3]{\eps}.\label{ineq:eta_1}
    \end{align}
\end{enumerate}

\paragraph{Computing the total distance between the input and output terms for each term.}
For each input Hamiltonian term $h_{i,j}$, we compute an upper bound on its distance from the output term $\hat{h}_{i,j}$. By the triangle inequality, there are two distances to add: (1) the distance $\|h^\snp_{i,j} - \hat{h}_{i,j}\|$, which we showed above is bounded by $192\ \eps^{1/3}$, and (2) the distance $\|h_{i,j} - h^\snp_{i,j} \|$, which we bounded earlier by $4 (\eta_2 + \eta_1)$. Thus, the total distance is bounded above as follows.
\begin{align*}
    \|\hat{h}_{i,j} - h_{i,j}\| &\leq \|\hat{h}_{i,j} - h^\snp_{i,j}\| + \|h^\snp_{i,j} - h_{i,j}\|\\
    & \leq 192\ \eps^{1/3} + 4 \eta_2 + 4 \eta_1\\
    & \leq 192\ \eps^{1/3} + 4 \sqrt[3]{\eps} + 4 \eps^{1/3} & \text{Ineq.~\ref{ineq:eta_2} and Ineq.~\ref{ineq:eta_1}}\\
    & \leq 200\ \eps^{1/3}.
\end{align*}
From this, the total operator norm bound
\[ \|\hat{H} - H \| \leq 200 m \eps^{1/3} \]
follows directly by applying the triangle inequality $m$ times.

\paragraph{Runtime:} The fact that the runtime is linear in $m$ follows by inspecting the algorithm: only the at most $2m$ qubits occurring in the input need to be processed, and all matrix operations are on constant-size matrices. For rational input entries and rational $\eps$, the gap comparisons involve algebraic numbers of bounded degree and polynomial bit complexity. The output terms can themselves be represented exactly with rational entries: midpoint snapping is rational, and for a non-scalar rational pivot $R$, writing $B=R-\frac{\mathrm{Tr}(R)}2I$ and $b^2=\mathrm{Tr}(B^2)/2>0$ gives
\[
\mathcal{P}_R(X)=\frac12\left(X+\frac{BXB}{b^2}\right).
\]
Each output term uses only a constant number of these rational matrix operations, so its bit complexity is polynomial in that of the input. In particular, exact commutation does not require rounding the output matrices to floating-point numbers.
\end{proof}

\begin{lemma}[The multi-qubit pinching lemma]\label{lemma:double-extra-cheese}
    Let $H = \sum_{(i,j) \in E} h_{ij}$ be a 2-local Hamiltonian and for each $i \in [n]$, let $\mathcal{P}_i$ be either the identity superoperator or the pinching superoperator for some one-local Hermitian operator $R_i$ acting on qubit $i$. Suppose that for each $i$ where $\mathcal{P}_i$ is not identity, $\Delta(R_i) \geq \eta_i$ and for all $j$, 
    \[ \|[R_i \otimes I_j, h_{i,j} ] \| \leq \kappa_i,\]
    for parameters $\eta_i>0$ and $\kappa_i\geq 0$; when $\mathcal{P}_i=\mathrm{id}$, set $\eta_i=1$ and $\kappa_i=0$.
    Then the 2-local Hamiltonian
    \begin{equation}\label{eqn:definition-of-H-prime}
        H' = \sum_{(i,j) \in E} (\cP_i \otimes \cP_j \otimes \mathrm{id}_{[n]\setminus \{i,j\}})( h_{i,j})
    \end{equation}
    satisfies the following properties:
    \begin{enumerate}
        \item For all $i$ such that $\cP_i$ is not identity, all the terms $h'_{i,j}$ touching qubit $i$ commute with each other, and with $R_i$.
        \item For all $(i,j) \in E$,
        \begin{equation} \| h'_{i,j} - h_{i,j} \| \leq 4 \left(\frac{\kappa_i}{\eta_i} +  \frac{\kappa_j}{\eta_j}\right). \label{eq:cheese-rounding-error} \end{equation}
    \end{enumerate}
\end{lemma}
\begin{proof}

We show the two properties in turn.
\begin{enumerate}
    \item \textbf{Commutativity.} We only need to check commutativity for pairs of terms which overlap on exactly one qubit (there are no two terms that overlap on two qubits since our convention\footnote{This convention is without loss of generality, since one can absorb all terms that act on a pair $i,j$ into a single term at the start of the analysis.} is that $h_{i,j}$ is the sole term in the Hamiltonian acting on the pair $i,j$), and that were affected nontrivially by the pinching on the shared qubit. Let $\{i, j, k\}$ be three distinct indices. Consider the terms $h_{i,j}$ and $h_{j,k}$ which have the following decomposition according to \Cref{algebra-fact-pauli-decomposition}.

    \begin{equation}\label{eqn:extra-cheese-decomp}
        h_{i,j} = \sum\limits_{\alpha = 0}^3 \sigma^\alpha_i \otimes A_j^\alpha \quad\quad\quad\quad h_{j,k} = \sum\limits_{\beta = 0}^3 B_j^\beta \otimes \sigma^\beta_k
    \end{equation}
    
Now, consider the commutator $[h'_{i,j}, h'_{j,k}]$.
\begin{align*}
    [h'_{i,j}, h'_{j,k}] &= [ (\cP_i \otimes \cP_j)(h_{i,j}), (\cP_j \otimes \cP_k)(h_{j,k}) ]  \qquad \text{(by \Cref{eqn:definition-of-H-prime})}\\
    &= \left[ \sum_{\alpha} (\cP_i \otimes \cP_j)(\sigma^\alpha_i \otimes A_j^\alpha), \sum_{\beta} (\cP_j \otimes \cP_k)(B_j^\beta \otimes \sigma^\beta_k) \right ]  \qquad \text{(by \Cref{eqn:extra-cheese-decomp})}\\
    &= \sum_{\alpha,\beta} \cP_i(\sigma_i^\alpha) \otimes \underbrace{[\cP_j(A^\alpha_j), \cP_j(B^\beta_j)]}_{=0} \otimes \cP_k(\sigma_k^\beta) \qquad \text{(commutator properties)}\\
    &= 0.
\end{align*}
To see that the commutators in the sum are $0$, observe that by \Cref{thm:pinch}, each $\cP_j(A_j^\alpha)$ and $\cP_j(B_j^\beta)$ commutes with $R_j$, and thus by \Cref{lem:transitivity} and the fact that $R_j$ is gapped, we get that $\cP_j(A_j^\alpha)$ and $\cP_j(B_j^\beta)$ commute with each other. The same calculation shows that each $h'_{i,j}$ commutes with $R_i$ whenever $\mathcal{P}_i$ is not identity. 

\item \textbf{Rounding error.} To bound the rounding error, we will first show a commutativity property of the operators $R_i$ for every qubit where they are defined. Consider such a qubit $i$ and let $h_{i,j}$ be a Hamiltonian term acting on qubit $i$. Decompose $h_{i,j}$ as
\begin{equation*}\tag{\ref*{eqn:extra-cheese-decomp}}
    h_{i,j}  = \sum_{\alpha=0}^3 A_i^\alpha \otimes \sigma_j^\alpha. 
\end{equation*}
Since, by the hypothesis, $\|[h_{i,j},(R_i \otimes I_j) ]\| \leq \kappa_i$, expanding this commutator and applying \Cref{lem:opnorm-orthog-components} shows that for each $\alpha$:
\begin{equation}
 \| [A_i^\alpha, R_i] \| \leq \kappa_i. \label{eq:R-commutator}
\end{equation}

Now, we can compute the rounding error of any Hamiltonian term. Let $h_{i,j}$ be any term of the Hamiltonian. If both $\cP_i$ and $\cP_j$ are equal to identity, then there is nothing to prove.

Now, suppose without loss of generality that qubit $i$ has $\cP_i$ not equal to identity, so $R_i$ is defined. Then we obtain
\begin{align}
    \| \cP_i(h_{i,j}) - h_{i,j} \| &= \left\|\sum_{\alpha =0}^3 \big(\cP_i(A_i^\alpha)-A_i^\alpha\big)\otimes\sigma_j^\alpha\right\| \\
    &\leq \sum_{\alpha =0}^3 \| \cP_i(A_i^\alpha) - A_i^\alpha\| \\
    &\leq  \frac{1}{\Delta(R_i)} \sum_{\alpha =0}^3 \|[A_i^\alpha, R_i] \| \\
    &\leq \frac{4}{\Delta(R_i)} \cdot \kappa_i \leq \frac{4 \kappa_i}{\eta_i},
 \end{align}
 where in the second to last line we applied \Cref{thm:pinch} and \Cref{eq:R-commutator}. Now, if $\cP_j$ is equal to identity, then $h'_{i,j} = \cP_i(h_{i,j})$, so we have shown \Cref{eq:cheese-rounding-error} in this case.

The remaining case is if $\cP_j$ is also not equal to identity. In this case, write the decomposition
\[ h_{i,j} = \sum_{\alpha=0}^{3} \sigma_i^\alpha \otimes C_j^\alpha. \]
We can similarly derive the bound
\begin{align}
    \| (\cP_i \otimes \cP_j)(h_{i,j}) - \cP_i(h_{i,j}) \| &= \| \sum_{\alpha =0}^3 \cP_i(\sigma_i^\alpha) \otimes (\cP_j(C_j^\alpha) - C_j^\alpha) \| \\
    &\leq  \sum_{\alpha =0}^3 \| (\cP_j(C_j^\alpha) - C_j^\alpha) \| \\
    &\leq \frac{1}{\Delta(R_j)} \sum_{\alpha =0}^3 \| [C_j^\alpha, R_j] \| \\
    &\leq \frac{4\kappa_j}{\eta_j} \qquad \text{(by \Cref{eq:R-commutator} and } \Delta(R_j) \geq \eta_j),
\end{align}
where we have used the bound $\| \cP_i(\sigma_i^\alpha) \| \leq \| \sigma_i^\alpha \| = 1$, which follows from the fact that $\cP_i$ is a positive unital map and thus contractive in operator norm (see~\cite[Theorem~6.12]{Watrous_2018}). 
So overall, by the triangle inequality, we get
\begin{align}
    \| h_{i,j} - (\cP_i \otimes \cP_j)(h_{i,j})  \| &\leq \| h_{i,j} - \cP_i(h_{i,j}) \| + \| \cP_i(h_{i,j}) - (\cP_i \otimes \cP_j)(h_{i,j}) \| \\
    &\leq \frac{4\kappa_i}{\eta_i} + \frac{4\kappa_j}{\eta_j}.
\end{align}
This proves \Cref{eq:cheese-rounding-error}.
\end{enumerate}
\end{proof}

\section{Applications to Almost Commuting $2$-local Hamiltonians}
In this section, we use our locality-preserving rounding technique (formally, \Cref{thm:transforming-hamiltonian}) to show that the almost commuting $2$-local Hamiltonian problem lies in $\mathsf{NP}$.

\subsection{Containment in $\NP$}
\begin{theorem}\label{thm:aclh-to-clh-reduction}
There exists a linear-time computable reduction from the $(\gamma, \eps)$-$2$-$\mathrm{ACLH}$ problem to the $(\gamma')$-$2$-$\mathrm{CLH}$ where $\gamma' = \gamma - 400\ \eps^{1/3}$ for any $0 \leq \eps \leq 1$ such that $\gamma'\geq 1/\poly(n)$.
\end{theorem}
\begin{proof}
Given $(H,a,b)$ an instance of $(\gamma,\eps)$-2-ACLH, let $H'$ be the commuting Hamiltonian obtained by applying \Cref{thm:transforming-hamiltonian} to $H$, and let $a' = a + 200m\eps^{1/3}$ and $b' = b - 200m\eps^{1/3}$. These thresholds have exact polynomial-bit descriptions as constant-degree algebraic numbers (a defining polynomial together with an isolating interval); their numerical approximations can be computed to any inverse-polynomial accuracy in polynomial time. From the operator norm bound in the theorem, together with \Cref{fact:weyl-corollary}, we know that \[|\lambda_{\min}(H') - \lambda_{\min}(H)| \leq 200 m \eps^{1/3}. \]

Hence, we have
\begin{align}
    \lambda_{\min}(H) \leq a &\Longrightarrow \lambda_{\min}(H') \leq a' \\
    \lambda_{\min}(H) \geq b &\Longrightarrow \lambda_{\min}(H') \geq b'.
\end{align}
Thus, since $(b'-a')/m=\gamma'\geq 1/\poly(n)$, it holds that $(H',a',b')$ is always an instance of $(\gamma')$-2-CLH satisfying the promise for $\gamma' = \gamma - 400\eps^{1/3}$, and $H'$ is a YES instance iff $H$ is a YES instance, so the map $H \mapsto H'$ is a reduction as claimed.
\end{proof}

\begin{corollary}
\label{corollary:NP-Containment}
    The $(\gamma, \eps)$-$2$-$\mathrm{ACLH}$ promise problem is in $\NP$ when $\gamma - 400\eps^{1/3} \geq 1/\poly(n)$.
\end{corollary}
\begin{proof}
This is a direct consequence of \Cref{thm:aclh-to-clh-reduction} and \Cref{thm:clh-in-np}.
\end{proof}

\subsection{Quantum Gibbs sampling}\label{sec:Gibbs}

As an application of our locality-preserving rounding technique for almost commuting quantum many-body Hamiltonians, we now consider the task of quantum Gibbs sampling.
For a local Hamiltonian $H$ and an inverse temperature $\beta > 0$, the Gibbs state 
\begin{align}\label{eq:Gibbs-state}
\rho_\beta(H) = \frac{e^{-\beta H}}{\mathrm{Tr}\left[e^{-\beta H}\right]}    
\end{align}
describes the thermal equilibrium properties of a quantum system at finite temperature $T=1/\beta$.
Given a precision parameter $\delta >0$, the $(H,\beta,\delta)$-Gibbs state preparation (or sampling) task with respect to the Hamiltonian $H$ and an inverse temperature $\beta > 0$ is to output $\rho$ such that
\begin{align}
\| \rho - \rho_\beta(H) \|_1 \leq \delta.    
\end{align}

\paragraph{Almost-commuting Hamiltonians.}

Our rounding technique allows us to reduce quantum Gibbs sampling for 
$\eps$-almost-commuting $2$-local qubit Hamiltonians
to the task of Gibbs sampling for a nearby commuting $2$-local qubit Hamiltonian, provided that $\eps$ is sufficiently small.

We use the following result on the continuity of the Gibbs state (with respect to the Hamiltonian) which is implicit in~\cite[Proposition 15]{Capel_2025}.
\begin{proposition}[Continuity of the Gibbs state,~\cite{Capel_2025}]\label{prop:continuity}
Let $H,\hat{H}$ be Hermitian operators acting on a finite-dimensional complex vector space. Then, for any inverse temperature $\beta > 0$,
$$
\| \rho_\beta(H) - \rho_\beta(\hat{H})\|_1 \, \leq \, e^{2 \beta \| H - \hat{H} \|} -1.
$$
In particular, if $\|H-\hat{H}\| \leq \ln(1+\delta)/2\beta$, then it holds that $\|\rho_\beta(H) - \rho_\beta(\hat{H})\|_1 \leq \delta$.
\end{proposition}

Equipped with the above technical fact, we can now show the following reduction.

\begin{theorem} 
Let $H = \sum_{i=1}^m h_i$ be an $\eps$-almost-commuting $2$-local qubit Hamiltonian with $\|h_i\|\leq 1$ and $\| [h_i,h_j]\| \leq \eps$, for all $i,j$, where $0\leq\eps\leq 1$. Then, for any inverse temperature $\beta >0$ and any precision $\delta >0$ such that
\begin{align}\label{eq:paramrters}
200 m \eps^{1/3} \leq  \ln\Big(1+\frac{\delta}{2} \Big)/2\beta
\end{align}
there exists a reduction from the task of $(H,\beta,\delta)$-Gibbs sampling to the task of $(\hat{H},\beta,\delta/2)$-Gibbs sampling for a $2$-local commuting qubit Hamiltonian $\hat{H}$. Moreover, if the Gibbs sampler for $\hat{H}$ runs in time $T$, then the resulting Gibbs sampler for $H$ runs in time $T + O(m)$.
\end{theorem}
\begin{proof}
The reduction simply applies the rounding technique from \Cref{thm:transforming-hamiltonian} to $H$, resulting in a $2$-local commuting qubit Hamiltonian $\hat{H}$ such that $\| H - \hat{H} \| \leq 200 m \eps^{1/3}$. Then, the output of any successful $(\hat{H},\beta,\delta/2)$-Gibbs sampler results in a state $\rho$ such that
\begin{align*}
\| \rho - \rho_\beta(H) \|_1 
&\leq \| \rho - \rho_\beta(\hat{H}) \|_1  + \| \rho_\beta(\hat{H}) - \rho_\beta(H) \|_1    && (\text{triangle inequality})\\
&\leq \frac{\delta}{2} + \| \rho_\beta(\hat{H}) - \rho_\beta(H) \|_1 && (\text{by assumption})\\
&\leq  \frac{\delta}{2}  + e^{2 \beta \| H - \hat{H} \|} -1  && (\text{\Cref{prop:continuity}})\\
&\leq \frac{\delta}{2} + \frac{\delta}{2}  = \delta  && (\text{due to~\eqref{eq:paramrters}})
\end{align*}
Because the rounding step is deterministic and requires time $O(m)$, the claim follows.
\end{proof}
We remark that our Gibbs sampling reduction for almost-commuting Hamiltonians above can be further combined with the reduction of~\cite[Theorem 9]{hwang2025gibbsstatepreparationcommuting}, thereby reducing Gibbs sampling for (certain) almost-commuting $2$-local qubit Hamiltonians to Gibbs sampling of certain classical Hamiltonians.

\subsection{Faster Hamiltonian simulation}\label{sec:hamiltonian-simulation}
Our result gives a method to speed up Hamiltonian simulation for approximately commuting Hamiltonians. The basic idea is to use the fact that time evolution of commuting Hamiltonians can be simulated very efficiently, since the matrix exponential factors into a product of exponentials of the individual terms. For short times, thus, one could simply round an approximately commuting Hamiltonian $H$ to a commuting Hamiltonian $H'$ with our technique, and simulate $H'$ instead of $H$. For longer times, the error incurred by this rounding may be prohibitive: for this more general case, we can use the "interaction picture" formalism of Low and Wiebe~\cite{low2018hamiltonian}, in essentially a black-box way, to add corrections from the rounding error. 

As a first step, we show the following helper lemma which is essentially folklore.
\begin{lemma}\label{lem:com-ham-fast}
    Suppose $H = \sum_{i=1}^{m} h_i$ is an $O(1)$-local commuting Hamiltonian on $n$ qubits, with $\|h_i\|\leq 1$, $m = \Omega(n)$, $t\in\mathbb{R}$, and $0<\delta<1$. Then there exists a unitary quantum circuit on $n$ system qubits and one ancillary qubit of gate complexity 
     \[ O\Big(m \, \poly\big(\log(m/\delta), \log(2+|t|)\big)\Big)\]
     that implements a unitary $U$ such that, with the ancillary qubit initialized to $\ket{0}$,
     \[ \big\|U(I\otimes\ket{0})-(e^{iHt}\otimes\ket{0})\big\| \leq \delta. \] 
     Moreover, this circuit can be computed from a classical description of $H$ and $t$ in classical time scaling as $O(m \poly(L,\log(m/\delta),\log(2+|t|)))$, where $L$ bounds the bit complexity of each local term and of $t$.
\end{lemma}
\begin{proof}
   Using the commutativity of the terms of $H$, we may write the exponential of $H$ as a product of exponentials of individual terms as follows:
    \[ e^{iHt} = e^{i\sum_{i=1}^{m} h_i t} = \prod_{i=1}^{m} e^{i h_i t}. \]
    So to simulate the time evolution, it suffices to simulate $e^{i h_i t}$ for each term $h_i$.
    To do this, since $h_i$ acts locally on a constant number of qubits, we can explicitly compute $e^{ih_i t}$, up to some appropriate finite precision, and then apply the Solovay-Kitaev theorem, in the algorithmic form of~\cite{dawson2006solovay}, to simulate it using a sequence of elementary gates. To retain scalar phases, embed each local unitary $V$ as $V\otimes\proj{0}+V^\dagger\otimes\proj{1}$ on one additional qubit. This embedded unitary has determinant one and acts as $V$ when the ancillary qubit is $\ket{0}$. The same ancillary qubit can be reused throughout, and the approximation errors accumulate by the triangle inequality.
    
    We now estimate the time complexity of this procedure. In order to implement $e^{iHt}$ up to error $\delta$, we need to implement each $e^{i h_i t}$ up to error $\delta/m$. To do this, if each term $h_i$ is $k$-local, it suffices to compute each entry of the matrix $e^{i h_i t}$ up to error $\delta/(100m\,2^k)$ (so that the total computed matrix is at most $\delta/(100m)$ far in operator norm from $e^{ih_i t}$). Replacing this numerical matrix by its unitary polar factor increases the error by at most a constant factor. We then apply the Solovay-Kitaev theorem to approximate its determinant-one embedding up to error $\delta/(100m)$. Applying the results of~\cite{dawson2006solovay}, we get the claimed gate complexity and runtime.
\end{proof}

We will now sketch how to use our rounding to perform Hamiltonian simulation. The main technical obstacle to directly applying a standard algorithm is the following: Hamiltonian simulation algorithms are typically given in the oracle model, where access to $H$ is given through a unitary oracle $O_H$ that implements a "block-encoding" of $H$: $(\bra{0}\otimes I)O_H (\ket{0} \otimes I) \propto H$. However, in our setting, we need to work in a model where the Hamiltonian is given as \emph{explicit} input (i.e. as a classical string containing descriptions of all the local terms), in order to apply our rounding algorithm. Indeed, implicit in the oracle picture is the assumption that for any realistic Hamiltonian, such an oracle can be implemented efficiently. So we will assume that we have a method to construct such a block encoding and leave its size as a parameter.

    Our rounding says that
    \[ H = H' + \Delta,\]
    where $H'$ is a commuting Hamiltonian with $\|H'\|\leq\alpha_A:=m$, and $\Delta$ is a local Hamiltonian with $\|\Delta\|\leq 200m\eps^{1/3}$. Here we assume the hypotheses of \Cref{thm:transforming-hamiltonian}. Moreover, $H'$ and $\Delta$ are efficiently computable from $H$. If $\Delta=0$, we use \Cref{lem:com-ham-fast} alone. Otherwise, choose a known normalization $\alpha_B>0$ such that
    \[ \|\Delta\|\leq\alpha_B\leq200m\eps^{1/3};\]
    for example, one may take $\alpha_B=200m\eps^{1/3}$. Suppose that, once we have written down $\Delta$, we can construct a block encoding of $\Delta/\alpha_B$, together with its controlled version and inverse, with gate complexity $T_{block}$ at the precision required below. We assume that the number of block-encoding ancillas is $O(T_{block})$, so reflections on those ancillas are included in the cost bound.

    We use the interaction-picture construction of Low and Wiebe~\cite{low2018hamiltonian}, with $H'$ as the ``fast'' part and $\Delta$ as the slow part. To account for small times explicitly, let $t>0$, $0<\delta<1$, and set
    \[
    r=\max\{1,\lceil2\alpha_B t\rceil\},\qquad \tau=t/r,\qquad
    H_I(s)=e^{iH's}\Delta e^{-iH's}\quad(0\leq s\leq\tau).
    \]
    Then
    \[
    e^{-iHt}=\left(e^{-iH'\tau}\,\mathcal{T}e^{-i\int_0^\tau H_I(s)\,\mathrm{d}s}\right)^r,
    \qquad
    \|H_I(s)\|\leq\alpha_B,\quad
    \|\dot H_I(s)\|\leq2\alpha_A\alpha_B.
    \]
    The truncated-Dyson-series algorithm of~\cite[Theorem~3 and Lemma~6]{low2018hamiltonian}, with error $O(\delta/r)$ per segment, uses $\widetilde O(r)$ queries in total to a time-indexed block encoding of $H_I(s)/\alpha_B$, together with $r$ evolutions under $H'$.

    For completeness, the required time-indexed oracle can be implemented directly, without assuming a simulation cost proportional to the duration of each short evolution. As in the proof of~\cite[Theorem~7]{low2018hamiltonian}, let
    \[
    V=\sum_{\ell=0}^{M-1}|\ell\rangle\langle\ell|\otimes I_a\otimes e^{-iH'\ell\tau/M},
    \]
    where $I_a$ acts on the block-encoding ancillas. If $O_\Delta$ block-encodes $\Delta/\alpha_B$, then $V^\dagger(I\otimes O_\Delta)V$ is the required time-indexed block encoding. Here $M$ can be chosen to be a power of two of order $1+t(\alpha_A+\alpha_B)/\delta+\log^2(2+r/\delta)$, as required by~\cite[Lemma~5]{low2018hamiltonian}. Expanding $\ell$ in binary implements this controlled evolution using $O(\log M)$ controlled evolutions under $H'$, each of duration at most $\tau\leq t$. The proof of \Cref{lem:com-ham-fast} applies to these controlled evolutions as well: controlling each local exponential increases its locality by only one. Each such evolution therefore costs $\widetilde O(m)$ gates, uniformly in its duration, when its error is chosen sufficiently small that all oracle-synthesis errors sum to $O(\delta)$. Thus a time-indexed oracle costs $\widetilde O(T_{block}+m)$ gates. The $r$ separate evolutions under $H'$ and the auxiliary gates in the Dyson-series algorithm fit within the same total bound.

    Putting this together and ignoring polylogarithmic factors gives gate complexity
    \[
    \widetilde O\big((1+\alpha_B t)(T_{block}+m)\big),
    \qquad\text{and hence}\qquad
    \widetilde O\big((1+m\eps^{1/3}t)(T_{block}+m)\big),
    \]
    where the hidden factors are polynomial in $\log(2+m(1+t)/\delta)$. The case $t=0$ requires no gates, and negative times are handled by taking the inverse circuit. Apart from the commuting-simulation overhead and polylogarithmic factors, the important feature is that the time-dependent cost scales with a normalization bound for $\Delta$, controlled by the commutator error, rather than with the total norm of the Hamiltonian.
\printbibliography

@article{10.1007/BF01994876,
  author = {Boppana, Ravi and Halld\'{o}rsson, Magn\'{u}s M.},
  title = {Approximating maximum independent sets by excluding subgraphs},
  year = {1992},
  issue_date = {Jun 1992},
  publisher = {BIT Computer Science and Numerical Mathematics},
  address = {USA},
  volume = {32},
  number = {2},
  issn = {0006-3835},
  url = {https://doi.org/10.1007/BF01994876},
  doi = {10.1007/BF01994876},
  journal = {BIT},
  month = jun,
  pages = {180–196},
  numpages = {17},
}

@article{Davidson1985,
  author = {Kenneth R. Davidson},
  title = {Almost commuting {Hermitian} matrices},
  journal = {Mathematica Scandinavica},
  volume = {56},
  pages = {222--240},
  year = {1985},
  publisher = {Department of Mathematics, Aarhus University},
  doi = {10.7146/math.scand.a-12098},
  url = {https://www.mscand.dk/article/view/12098},
}

@misc{herrera2022hastingsapproachlinstheorem,
  title = {On Hastings' approach to Lin's Theorem for Almost Commuting Matrices},
  author = {David Herrera},
  year = {2022},
  eprint = {2011.11800},
  archiveprefix = {arXiv},
  primaryclass = {math.FA},
  url = {https://arxiv.org/abs/2011.11800},
  doi = {10.48550/arXiv.2011.11800},
}

@article{Lin1997,
  author = {Lin, Huaxin},
  title = {Almost commuting selfadjoint matrices and applications},
  journal = {Fields Institute Communications},
  volume = {13},
  pages = {193--233},
  year = {1997},
  publisher = {American Mathematical Society},
  address = {Providence, RI},
  doi = {10.1090/fic/013/11},
}

@book{Kato:1966:PTL,
  at = {2006-05-12T17:24:33.000+0200},
  alias = {Kato 66},
  author = {Kato, Tosio},
  bibdate = {Fri Nov 24 15:18:30 1995},
  bibsource = {ftp://ftp.math.utah.edu/pub/bibnet/authors/m/matched-field-proc.bib},
  key = {eigenvalues, hilbert spaces},
  lccn = {QA320 .K33},
  pages = {xix + 592},
  sthbib = {M3 Kat 81 60},
  title = {Perturbation Theory for Linear Operators},
  year = {1966},
  doi = {10.1007/978-3-662-12678-3},
  publisher = {Springer},
}

@article{Hastings_2010,
  title = {Almost commuting matrices, localized Wannier functions, and the quantum Hall effect},
  volume = {51},
  issn = {1089-7658},
  url = {https://doi.org/10.1063/1.3274817},
  doi = {10.1063/1.3274817},
  number = {1},
  journal = {Journal of Mathematical Physics},
  publisher = {AIP Publishing},
  author = {Hastings, Matthew B. and Loring, Terry A.},
  year = {2010},
  month = jan,
}

@article{arad2011notepartialnogotheorem,
  title = {A note about a partial no-go theorem for quantum {PCP}},
  author = {Arad, Itai},
  journal = {Quantum Information \& Computation},
  volume = {11},
  number = {11-12},
  pages = {1019--1027},
  year = {2011},
  publisher = {Rinton Press, Incorporated Paramus, NJ},
  url = {https://www.rintonpress.com/xxqic11/qic-11-1112/1019-1027.pdf},
  doi = {10.26421/QIC11.11-12-10},
}

@article{Mbeng_2024,
  title = {The quantum Ising chain for beginners},
  issn = {2590-1990},
  url = {https://doi.org/10.21468/scipostphyslectnotes.82},
  doi = {10.21468/scipostphyslectnotes.82},
  journal = {SciPost Physics Lecture Notes},
  publisher = {Stichting SciPost},
  author = {Mbeng, Glen Bigan and Russomanno, Angelo and Santoro, Giuseppe E.},
  year = {2024},
  month = jun,
}

@article{Vijay_2016,
  title = {Fracton topological order, generalized lattice gauge theory, and duality},
  volume = {94},
  issn = {2469-9969},
  url = {https://doi.org/10.1103/physrevb.94.235157},
  doi = {10.1103/physrevb.94.235157},
  number = {23},
  journal = {Physical Review B},
  publisher = {American Physical Society (APS)},
  author = {Vijay, Sagar and Haah, Jeongwan and Fu, Liang},
  year = {2016},
  month = dec,
}

@article{Levin_2005,
  title = {String-net condensation: A physical mechanism for topological phases},
  volume = {71},
  issn = {1550-235X},
  url = {https://doi.org/10.1103/physrevb.71.045110},
  doi = {10.1103/physrevb.71.045110},
  number = {4},
  journal = {Physical Review B},
  publisher = {American Physical Society (APS)},
  author = {Levin, Michael A. and Wen, Xiao-Gang},
  year = {2005},
  month = jan,
}

@misc{chen2023fastthermalizationeigenstatethermalization,
  title = {Fast Thermalization from the Eigenstate Thermalization Hypothesis},
  author = {Chi-Fang Chen and Fernando G. S. L. Brandão},
  year = {2023},
  eprint = {2112.07646},
  archiveprefix = {arXiv},
  primaryclass = {quant-ph},
  url = {https://arxiv.org/abs/2112.07646},
  doi = {10.48550/arXiv.2112.07646},
}

@article{kastoryano2016quantumgibbssamplerscommuting,
  title = {Quantum Gibbs Samplers: the commuting case},
  author = {Michael J. Kastoryano and Fernando G. S. L. Brandao},
  year = {2016},
  eprint = {1409.3435},
  archiveprefix = {arXiv},
  primaryclass = {quant-ph},
  url = {https://arxiv.org/abs/1409.3435},
  doi = {10.1007/s00220-016-2641-8},
  journal = {Communications in Mathematical Physics},
  volume = {344},
  number = {3},
  pages = {915--957},
}

@inproceedings{Anshu_2023,
  series = {STOC ’23},
  title = {NLTS Hamiltonians from Good Quantum Codes},
  url = {https://doi.org/10.1145/3564246.3585114},
  doi = {10.1145/3564246.3585114},
  booktitle = {Proceedings of the 55th Annual ACM Symposium on Theory of Computing},
  publisher = {ACM},
  author = {Anshu, Anurag and Breuckmann, Nikolas P. and Nirkhe, Chinmay},
  year = {2023},
  month = jun,
  pages = {1090–1096},
  collection = {STOC ’23},
}

@article{Hastings2012,
  title = {Trivial low-energy states for commuting Hamiltonians, and the quantum PCP conjecture},
  author = {Hastings, Matthew B.},
  journal = {Quantum Information and Computation},
  volume = {13},
  number = {5--6},
  pages = {393--429},
  year = {2013},
  note = {arXiv:1201.3387},
  doi = {10.26421/QIC13.5-6-3},
}

@article{Kitaev_2003,
  title = {Fault-tolerant quantum computation by anyons},
  volume = {303},
  issn = {0003-4916},
  url = {https://doi.org/10.1016/s0003-4916(02)00018-0},
  doi = {10.1016/s0003-4916(02)00018-0},
  number = {1},
  journal = {Annals of Physics},
  publisher = {Elsevier BV},
  author = {Kitaev, A.Yu.},
  year = {2003},
  month = jan,
  pages = {2–30},
}

@article{Gharibian_2015,
  title = {Quantum Hamiltonian Complexity},
  volume = {10},
  issn = {1551-3068},
  url = {https://doi.org/10.1561/0400000066},
  doi = {10.1561/0400000066},
  number = {3},
  journal = {Foundations and Trends® in Theoretical Computer Science},
  publisher = {Emerald},
  author = {Gharibian, Sevag and Huang, Yichen and Landau, Zeph and Shin, Seung Woo},
  year = {2015},
  pages = {159–282},
}

@article{irani2023commutinglocalhamiltonianproblem,
  title = {Commuting Local Hamiltonian Problem on 2D beyond qubits},
  author = {Sandy Irani and Jiaqing Jiang},
  year = {2025},
  eprint = {2309.04910},
  archiveprefix = {arXiv},
  primaryclass = {quant-ph},
  url = {https://arxiv.org/abs/2309.04910},
  doi = {10.1007/s00220-025-05462-8},
  journal = {Communications in Mathematical Physics},
  volume = {406},
  number = {11},
  eid = {273},
}

@book{PathriaBeale2011,
  title = {Statistical Mechanics},
  author = {Pathria, R. K. and Beale, Paul D.},
  edition = {3},
  year = {2011},
  publisher = {Elsevier},
  doi = {10.1016/C2009-0-62310-2},
}

@book{Chandler1987,
  title = {Introduction to Modern Statistical Mechanics},
  author = {Chandler, David},
  year = {1987},
  publisher = {Oxford University Press},
}

@online{tao2010eigenvalues,
  author = {Terence Tao},
  title = {254A, Notes 3a: Eigenvalues and sums of Hermitian matrices},
  year = {2010},
  url = {https://terrytao.wordpress.com/2010/01/12/254a-notes-3a-eigenvalues-and-sums-of-hermitian-matrices/},
  note = {Accessed: 2026-01-23},
  date = {2010-01-12},
}

@book{Callen1985,
  title = {Thermodynamics and an Introduction to Thermostatistics},
  author = {Callen, Herbert B.},
  edition = {2},
  year = {1985},
  publisher = {Wiley},
}

@book{LandauLifshitz1980,
  title = {Statistical Physics, Part 1},
  author = {Landau, L. D. and Lifshitz, E. M.},
  edition = {3},
  year = {1980},
  series = {Course of Theoretical Physics, Vol. 5},
  publisher = {Pergamon Press},
  doi = {10.1016/C2009-0-24487-4},
}

@book{Sachdev2011,
  title = {Quantum Phase Transitions},
  author = {Sachdev, Subir},
  edition = {2},
  year = {2011},
  publisher = {Cambridge University Press},
  doi = {10.1017/CBO9780511973765},
}

@book{AltlandSimons2010,
  title = {Condensed Matter Field Theory},
  author = {Altland, Alexander and Simons, Ben D.},
  edition = {2},
  year = {2010},
  publisher = {Cambridge University Press},
  doi = {10.1017/CBO9780511789984},
}

@article{Mermin1965,
  title = {Thermal Properties of the Inhomogeneous Electron Gas},
  author = {Mermin, N. David},
  journal = {Physical Review},
  year = {1965},
  volume = {137},
  number = {5A},
  pages = {A1441--A1443},
  doi = {10.1103/PhysRev.137.A1441},
}

@article{feng2022swendsenwangdynamicsferromagneticising,
  title = {Swendsen-Wang dynamics for the ferromagnetic Ising model with external fields},
  author = {Weiming Feng and Heng Guo and Jiaheng Wang},
  year = {2023},
  eprint = {2205.01985},
  archiveprefix = {arXiv},
  primaryclass = {cs.DS},
  url = {https://arxiv.org/abs/2205.01985},
  doi = {10.1016/j.ic.2023.105066},
  journal = {Information and Computation},
  volume = {294},
  pages = {105066},
}

@article{PhysRevLett.58.86,
  title = {Nonuniversal critical dynamics in Monte Carlo simulations},
  author = {Swendsen, Robert H. and Wang, Jian-Sheng},
  journal = {Phys. Rev. Lett.},
  volume = {58},
  issue = {2},
  pages = {86--88},
  numpages = {0},
  year = {1987},
  month = jan,
  publisher = {American Physical Society},
  doi = {10.1103/PhysRevLett.58.86},
  url = {https://link.aps.org/doi/10.1103/PhysRevLett.58.86},
}

@misc{schmidhuber2025hamiltoniandecodedquantuminterferometry,
  title = {Hamiltonian Decoded Quantum Interferometry},
  author = {Alexander Schmidhuber and Jonathan Z. Lu and Noah Shutty and Stephen Jordan and Alexander Poremba and Yihui Quek},
  year = {2025},
  eprint = {2510.07913},
  archiveprefix = {arXiv},
  primaryclass = {quant-ph},
  url = {https://arxiv.org/abs/2510.07913},
  doi = {10.48550/arXiv.2510.07913},
}

@article{hwang2025gibbsstatepreparationcommuting,
  title = {Gibbs state preparation for commuting Hamiltonian: Mapping to classical Gibbs sampling},
  author = {Yeongwoo Hwang and Jiaqing Jiang},
  year = {2026},
  eprint = {2410.04909},
  archiveprefix = {arXiv},
  primaryclass = {quant-ph},
  url = {https://arxiv.org/abs/2410.04909},
  doi = {10.22331/q-2026-09-18-2209},
  journal = {Quantum},
  volume = {10},
  pages = {2209},
}

@article{rajakumar2024gibbs,
  title = {Gibbs sampling gives quantum advantage at constant temperatures with {O(1)-local Hamiltonians}},
  author = {Rajakumar, Joel and Watson, James D},
  journal = {Quantum},
  year = {2026},
  doi = {10.22331/q-2026-01-22-1981},
  volume = {10},
  pages = {1981},
  eprint = {2408.01516},
  archiveprefix = {arXiv},
}

@inproceedings{BCL24,
  title = {Quantum computational advantage with constant-temperature {G}ibbs sampling},
  author = {Thiago Bergamaschi and Chi-Fang Chen and Yunchao Liu},
  booktitle = {2024 IEEE 65th Annual Symposium on Foundations of Computer Science (FOCS)},
  pages = {1063--1085},
  year = {2024},
  organization = {IEEE},
  note = {arXiv:2404.14639},
  doi = {10.1109/FOCS61266.2024.00071},
}

@article{Temme_2011,
  title = {Quantum {M}etropolis sampling},
  volume = {471},
  issn = {1476-4687},
  url = {https://doi.org/10.1038/nature09770},
  doi = {10.1038/nature09770},
  number = {7336},
  journal = {Nature},
  publisher = {Springer Science and Business Media LLC},
  author = {Temme, K. and Osborne, T. J. and Vollbrecht, K. G. and Poulin, D. and Verstraete, F.},
  year = {2011},
  month = mar,
  pages = {87–90},
}

@article{chen2023efficient,
  title = {An efficient and exact noncommutative quantum Gibbs sampler},
  author = {Chen, Chi-Fang and Kastoryano, Michael J and Gily{\'e}n, Andr{\'a}s},
  journal = {arXiv preprint arXiv:2311.09207},
  year = {2023},
  doi = {10.48550/arXiv.2311.09207},
  eprint = {2311.09207},
  archiveprefix = {arXiv},
}

@article{GilyenThermal23a,
  title = {Quantum Thermal State Preparation},
  author = {Chi-Fang Chen and Michael J. Kastoryano and Fernando G. S. L. Brand{\~a}o and Andr{\'a}s Gily{\'e}n},
  journal = {arXiv preprint},
  year = {2023},
  url = {https://arxiv.org/abs/2303.18224},
  doi = {10.48550/arXiv.2303.18224},
  eprint = {2303.18224},
  archiveprefix = {arXiv},
}

@book{Tuckerman2010,
  title = {Statistical Mechanics: Theory and Molecular Simulation},
  author = {Tuckerman, Mark E.},
  year = {2010},
  publisher = {Oxford University Press},
}

@book{LandauBinder2014,
  title = {A Guide to {M}onte {C}arlo Simulations in Statistical Physics},
  author = {Landau, David P. and Binder, Kurt},
  edition = {4},
  year = {2014},
  publisher = {Cambridge University Press},
  doi = {10.1017/CBO9781139696463},
}

@article{Capel_2025,
  title = {From Decay of Correlations to Locality and Stability of the Gibbs State},
  volume = {406},
  issn = {1432-0916},
  url = {https://doi.org/10.1007/s00220-024-05198-x},
  doi = {10.1007/s00220-024-05198-x},
  number = {2},
  journal = {Communications in Mathematical Physics},
  publisher = {Springer Science and Business Media LLC},
  author = {Capel, Ángela and Moscolari, Massimo and Teufel, Stefan and Wessel, Tom},
  year = {2025},
  month = jan,
}

@inproceedings{bostanci2025commutinglocalhamiltonians2d,
  title = {Commuting Local Hamiltonians Beyond 2D},
  author = {John Bostanci and Yeongwoo Hwang},
  year = {2026},
  eprint = {2410.10495},
  archiveprefix = {arXiv},
  primaryclass = {quant-ph},
  url = {https://arxiv.org/abs/2410.10495},
  doi = {10.4230/LIPIcs.ITCS.2026.25},
  booktitle = {17th Innovations in Theoretical Computer Science Conference (ITCS 2026)},
  series = {Leibniz International Proceedings in Informatics (LIPIcs)},
  volume = {362},
  pages = {25:1--25:20},
  publisher = {Schloss Dagstuhl -- Leibniz-Zentrum f{\"u}r Informatik},
}

@article{bravyi2004commutativeversionklocalhamiltonian,
  title = {Commutative version of the local Hamiltonian problem and common eigenspace problem},
  author = {S. Bravyi and M. Vyalyi},
  year = {2005},
  eprint = {quant-ph/0308021},
  archiveprefix = {arXiv},
  primaryclass = {quant-ph},
  url = {https://arxiv.org/abs/quant-ph/0308021},
  doi = {10.26421/QIC5.3-2},
  journal = {Quantum Information and Computation},
  volume = {5},
  number = {3},
  pages = {187--215},
}

@article{kkr,
  title = {The Complexity of the Local Hamiltonian Problem},
  author = {Julia Kempe and Alexei Kitaev and Oded Regev},
  year = {2006},
  eprint = {quant-ph/0406180},
  archiveprefix = {arXiv},
  primaryclass = {quant-ph},
  url = {https://arxiv.org/abs/quant-ph/0406180},
  doi = {10.1137/S0097539704445226},
  journal = {SIAM Journal on Computing},
  volume = {35},
  number = {5},
  pages = {1070--1097},
}

@inproceedings{aharonov2011complexity,
  title = {On the complexity of commuting local Hamiltonians, and tight conditions for topological order in such systems},
  author = {Aharonov, Dorit and Eldar, Lior},
  booktitle = {2011 IEEE 52nd Annual Symposium on Foundations of Computer Science},
  pages = {334--343},
  year = {2011},
  organization = {IEEE},
  archiveprefix = {arXiv},
  eprint = {1102.0770},
  doi = {10.1109/FOCS.2011.58},
}

@inproceedings{overlapping,
  doi = {10.4230/LIPIcs.ITCS.2017.48},
  url = {https://drops.dagstuhl.de/entities/document/10.4230/LIPIcs.ITCS.2017.48},
  author = {Chao, Rui and Reichardt, Ben W. and Sutherland, Chris and Vidick, Thomas},
  language = {en},
  title = {Overlapping Qubits},
  volume = {67},
  pages = {48:1-48:21},
  publisher = {Schloss Dagstuhl – Leibniz-Zentrum für Informatik},
  year = {2017},
  copyright = {Creative Commons Attribution 3.0 Unported license},
  booktitle = {8th Innovations in Theoretical Computer Science Conference (ITCS 2017)},
}

@article{low2018hamiltonian,
  title = {Hamiltonian simulation in the interaction picture},
  author = {Low, Guang Hao and Wiebe, Nathan},
  archiveprefix = {arXiv},
  eprint = {1805.00675},
  year = {2018},
  doi = {10.48550/arXiv.1805.00675},
}

@article{Hastings_2009,
  title = {Making Almost Commuting Matrices Commute},
  volume = {291},
  issn = {1432-0916},
  url = {https://doi.org/10.1007/s00220-009-0877-2},
  doi = {10.1007/s00220-009-0877-2},
  number = {2},
  journal = {Communications in Mathematical Physics},
  publisher = {Springer Science and Business Media LLC},
  author = {Hastings, M. B.},
  year = {2009},
  month = jul,
  pages = {321–345},
}

@article{dawson2006solovay,
  title = {The Solovay-Kitaev algorithm},
  author = {Dawson, Christopher M and Nielsen, Michael A},
  journal = {Quantum Information \& Computation},
  volume = {6},
  number = {1},
  pages = {81--95},
  year = {2006},
  publisher = {Rinton Press, Incorporated Paramus, NJ},
  archiveprefix = {arXiv},
  eprint = {quant-ph/0505030},
  doi = {10.26421/QIC6.1-6},
}

@article{schuch2011complexitycommutinghamiltonianssquare,
  title = {Complexity of commuting Hamiltonians on a square lattice of qubits},
  author = {Norbert Schuch},
  year = {2011},
  eprint = {1105.2843},
  archiveprefix = {arXiv},
  primaryclass = {quant-ph},
  url = {https://arxiv.org/abs/1105.2843},
  doi = {10.26421/QIC11.11-12-1},
  journal = {Quantum Information and Computation},
  volume = {11},
  number = {11--12},
  pages = {901--912},
}

@inproceedings{AKV,
  doi = {10.4230/LIPIcs.TQC.2018.2},
  url = {https://drops.dagstuhl.de/entities/document/10.4230/LIPIcs.TQC.2018.2},
  author = {Aharonov, Dorit and Kenneth, Oded and Vigdorovich, Itamar},
  language = {en},
  title = {On the Complexity of Two Dimensional Commuting Local Hamiltonians},
  volume = {111},
  pages = {2:1-2:21},
  publisher = {Schloss Dagstuhl – Leibniz-Zentrum für Informatik},
  year = {2018},
  copyright = {Creative Commons Attribution 3.0 Unported license},
  booktitle = {13th Conference on the Theory of Quantum Computation, Communication and Cryptography (TQC 2018)},
}

@book{Watrous_2018,
  place = {Cambridge},
  title = {The Theory of Quantum Information},
  publisher = {Cambridge University Press},
  author = {Watrous, John},
  year = {2018},
  doi = {10.1017/9781316848142},
}

@article{childs2021theory,
  title = {Theory of trotter error with commutator scaling},
  author = {Childs, Andrew M and Su, Yuan and Tran, Minh C and Wiebe, Nathan and Zhu, Shuchen},
  journal = {Physical Review X},
  volume = {11},
  number = {1},
  pages = {011020},
  year = {2021},
  publisher = {APS},
  doi = {10.1103/PhysRevX.11.011020},
}

@article{halmos1976some,
  title = {Some unsolved problems of unknown depth about operators on Hilbert space},
  author = {Halmos, Paul R},
  journal = {Proceedings of the Royal Society of Edinburgh Section A: Mathematics},
  volume = {76},
  number = {1},
  pages = {67--76},
  year = {1976},
  publisher = {Royal Society of Edinburgh Scotland Foundation},
  doi = {10.1017/S0308210500019491},
}

@article{pearcy1979almost,
  title = {Almost commuting matrices},
  author = {Pearcy, Carl and Shields, Allen},
  journal = {Journal of Functional Analysis},
  volume = {33},
  number = {3},
  pages = {332--338},
  year = {1979},
  publisher = {Elsevier},
  doi = {10.1016/0022-1236(79)90071-5},
}

\appendix
\section{Lower bound on the term-wise distance}
\label{sec:lower-bound}
We now show a concrete pair of $2$-qubit hermitian operators that $\varepsilon$-almost commute for which any exactly commuting approximation needs to move one of them by $\Omega(\sqrt{\epsilon})$.

\begin{theorem}[Lower bound of rounding two anticommuting binary observables]\label{thm:lower-bound}
Let $P,Q$ be a pair of anti-commuting ($PQ = - QP$) binary observables on a Hilbert space $\mathcal H$. Let $t > 0$. Then,
\begin{enumerate}
    \item $\varepsilon := \|[tP, tQ]\| = 2t^2$, and
    \item for any pair of commuting Hermitian operators $\hat{P}, \hat{Q}$, the following holds:
\begin{equation}
    \|tP - \hat{P}\|^2 + \|tQ - \hat{Q}\|^2 \geq t^2 = \varepsilon /2.
\end{equation}
\end{enumerate}

\noindent In other words, in any rounding scheme, one of the rounded operators must be moved by a distance of at least $\frac{\sqrt{\varepsilon}}{2}$.
\end{theorem}
\begin{proof}
Since $\hat{P}$ and $\hat{Q}$ commute, they are simultaneously diagonalizable, and thus have a common eigenvector $\ket{\psi}$. Thus,
\begin{equation}
    \hat{P} \ket{\psi} = \hat{p} \ket{\psi}, \quad\quad \hat{Q} \ket{\psi} = \hat{q} \ket{\psi}.
\end{equation}
Also, define $p$ and $q$ as follows:
\begin{equation}
    p := \bra{\psi} P \ket{\psi} \quad \quad \text{ and } \quad \quad  q := \bra{\psi} Q \ket{\psi}
\end{equation}
where $p$ and $q$ are real since $P$ and $Q$ are Hermitian. We now have:
\begin{align*}
 \|tP-\hat{P}\|_{\text{op}}^2
 &\ge\|(tP-\hat{P}) \ket{\psi}\|_2^2 & \text{by Definition of $\|\cdot\|_{\text{op}}$}\\
 &=\|tP\ket{\psi}- \hat{p} \ket{\psi}\|_2^2\\
 &= t^2 \bra{\psi} P^2 \ket{\psi} - 2t \hat{p} \bra{\psi} P \ket{\psi} + \hat{p}^2 \braket{\psi | \psi} \\
 &=t^2 - 2 t p \hat{p}+\hat{p}^2\\
 &= (\hat{p}-tp)^2+t^2(1-p^2)
 \ge t^2(1-p^2).
\end{align*}
Similarly, we have:
\begin{equation}
    \|tQ-\hat{Q}\|^2 \geq t^2 (1 - q^2).
\end{equation}

\noindent Thus,
\begin{equation}
    \|tP-\hat{P}\|^2 + \|tQ-\hat{Q}\|^2 \geq t^2(2 - (p^2 + q^2)).
\end{equation}

\noindent If $p^2 + q^2 \leq 1$, we obtain our result. We show that this is the case during the rest of this proof. Define 
\begin{equation}
    R := pP + qQ \quad\quad \text{ and }\quad\quad s := p^2 + q^2.
\end{equation}
By anticommutation ($PQ + QP = 0$) and since $P^2 = Q^2 = I$, we have:
\begin{align}
    R^2 &= (pP + qQ)^2\\
    &= p^2 P^2 + pq (PQ + QP) + q^2 Q^2\\
    &= p^2 \cdot I + 0 + q^2 \cdot I\\
    &= (p^2 + q^2) \cdot I= s \cdot I.
\end{align}

\noindent Since $R$ is hermitian, $\|R\|^2 = \|R^2\| = |s| = s$ (since $s = p^2 + q^2$ is non-negative). Thus, $\|R\| = \sqrt{s}$. On the other hand, we have:
\begin{equation}
    \braket{\psi | R |\psi} = \braket{\psi | pP |\psi} + \braket{\psi | qQ |\psi} = p^2 + q^2 = s.
\end{equation}
By the definition of operator norm, we have $s \leq \|R\| = \sqrt{s}$. Thus, $s \leq 1$. 
\end{proof}

\begin{corollary}
For $0<\eps\le1$, set
\begin{equation}\label{eq:gadget}
 h_{12}=tX_1X_2,\qquad h_{13}=tZ_1Z_3,
 \qquad t=\sqrt{\eps/2},
\end{equation}
and let $\hat{P}$ and $\hat{Q}$ be a pair of Hermitian commuting operators. Then, $\|[h_{12}, h_{13}]\| = \varepsilon$, but:
\begin{equation}
  \|h_{12} - \hat{P}\| \geq \frac{\sqrt{\varepsilon}}{2}  \quad\quad \text{ or } \quad\quad \|h_{13} - \hat{Q}\| \geq \frac{\sqrt{\varepsilon}}{2}.
\end{equation}
\end{corollary}

\begin{proof}
Since $X_1$ and $Z_1$ anti-commute, and since the Pauli matrices are Hermitian observables, the result follows directly by \Cref{thm:lower-bound}.
\end{proof}

\end{document}